\documentclass[10pt,twocolumn,twoside]{IEEEtran} 
\IEEEoverridecommandlockouts    
\usepackage{psfrag}
\usepackage{amsmath}
\usepackage{amssymb}
\usepackage{amsmath,amssymb,amsfonts}
\usepackage{amsfonts}
\usepackage{latexsym}
\usepackage{graphicx}
\usepackage{lettrine}
\usepackage{arydshln}
\usepackage{multirow}
\usepackage{accents}
\usepackage{mathtools}
\usepackage[subnum]{cases}
\usepackage{colortbl}
\usepackage{fancyhdr}
\usepackage{epstopdf}
\usepackage{float}
\usepackage{subfigure}
\usepackage{multirow}
\usepackage{rotating}
\IEEEoverridecommandlockouts
\usepackage{enumerate}
\usepackage{multicol}
\usepackage[ruled]{algorithm2e}
\usepackage[usenames,dvipsnames,svgnames,table]{xcolor}
\usepackage{soul}
\usepackage{verbatim}
\usepackage{hyperref}
\usepackage{pgfplots}
\pgfplotsset{compat=newest} 
\pgfplotsset{plot coordinates/math parser=false} 

\usepackage{caption}
\usepackage{chngcntr}
\usepackage{apptools}
\AtAppendix{\counterwithin{lemma}{subsection}}
\AtAppendix{\counterwithin{theorem}{subsection}}
\AtAppendix{\counterwithin{assumption}{subsection}}

\usepackage{amsthm}
\newtheorem{theorem}{Theorem}[]
\newtheorem{definition}{Definition}
\newtheorem{assumption}{Assumption}
\newtheorem*{assumption*}{Assumption}
\newtheorem*{sassumption*}{Standing assumption}
\newtheorem{proposition}{Proposition}
\newtheorem{lemma}{Lemma}

\newtheorem{remark}{Remark}

\DeclareSymbolFont{bbold}{U}{bbold}{m}{n}
\DeclareSymbolFontAlphabet{\mathbbold}{bbold}

\newcommand{\vect}[1]{\mathbbold{#1}}
\newcommand{\zeros}[1][]{\vect{0}_{#1}}
\newcommand{\ones}[1][]{\vect{1}_{#1}}

\newcommand{\argmin}[1]{\underset{#1}{\operatorname{argmin}}\,}
\newcommand{\infinf}[1]{\underset{#1}{\operatorname{inf}}\,}
\newcommand{\supsup}[1]{\underset{#1}{\operatorname{sup}}\,}

\DeclarePairedDelimiter{\floor}{\lfloor}{\rfloor}

\makeatletter
\newcommand{\specialcell}[1]{\ifmeasuring@#1\else\omit$\displaystyle#1$\ignorespaces\fi}
\makeatother
\newcommand{\smallblacksquare}{\mbox{\rule[0pt]{1.3ex}{1.3ex}}}

\renewcommand\i{^i}
\renewcommand\j{^j}
\newcommand{\aag}{\infty}
\newcommand{\R}{\mathbb{R}}
\newcommand{\Z}{\mathbb{Z}}
\newcommand{\N}{N}
\newcommand{\mc}{\mathcal}
\newcommand{\X}{\mc{X}}
\newcommand{\K}{\mc{K}}
\newcommand{\G}{\mc{G}}
\newcommand{\Ginf}{\mc{G}_\infty}

\newcommand{\maxx}[1]{\underset{#1}{\operatorname{max}}\,}
\renewcommand{\t}{\theta}
\renewcommand{\tt}{\hat \theta}

\newcommand{\Finf}{F_\aag}
\newcommand{\Qinf}{Q_\aag}
\newcommand{\Ainf}{A_\aag}
\newcommand{\binf}{b}
\newcommand{\sinf}{\sigma_\infty}
\newcommand{\di}{n}
\newcommand{\Gnu}{\G_\nu}
\newcommand{\Fnu}{F_\nu}
\newcommand{\Qnu}{Q_\nu}
\newcommand{\Anu}{A_\nu}

\newcommand{\snu}{\sigma_\nu}

\newcommand{\Sinf}{S_\infty}
\newcommand{\nuSMON}{\nu_{\textup{SMON}}}
\newcommand{\nueps}{\nu_{\varepsilon}}
\newcommand{\be}{\begin{equation}}
\newcommand{\ee}{\end{equation}}

\newcommand{\defeq}{\coloneqq}
\newcommand{\eqdef}{\eqqcolon}

\definecolor{Gray}{gray}{0.9}
\definecolor{Granata}{rgb}{0.64,0,0}
\definecolor{QuasiBlue}{rgb}{0.03,0.3,0.72} 
\definecolor{PortlandGreen}{RGB}{99,166,63} 
\definecolor{YellowIntense}{RGB}{240,198,31} 
\definecolor{OrangeRed}{RGB}{255,69,0} 
\definecolor{Granata}{rgb}{0.64,0,0} 
\definecolor{ReadableGreen}{RGB}{0,120,0}

\newcommand{\basi}{\textcolor{ReadableGreen}}

\title{\LARGE \bf
A distributed algorithm for average aggregative games \\ with coupling constraints}

\IEEEoverridecommandlockouts
\author{Francesca Parise, Basilio Gentile and  John Lygeros%
\thanks{The first two authors contributed equally. F. Parise is with the Laboratory for Information and Decision Systems, MIT, Cambridge, MA, USA. Email: \texttt{parisef}\texttt{@mit.edu}. B. Gentile and J. Lygeros  are with the Automatic Control Laboratory, ETH Zurich, Switzerland.
Email: \{\texttt{gentileb, lygeros}\}\texttt{@control.ee.ethz.ch}.
Research partially supported by the European Commission under project DYMASOS (FP7-ICT 611281) and by the Swiss National Science Foundation (grant  P2EZP2\_168812).
}}
\begin{document}
\setstcolor{PortlandGreen}
\setul{}{1.5pt}

\maketitle

\begin{abstract}
We consider the framework of average aggregative games, where the cost function of each agent depends on his own strategy and on the average population strategy.
We focus on the case in which the agents are coupled not only via their cost functions,
but also via constraints coupling their strategies.
We propose a distributed algorithm that achieves an almost-Nash equilibrium
by requiring only local communications of the agents, as specified by a sparse communication network.
The proof of convergence of the algorithm relies on the auxiliary class of network aggregative games and
exploits a novel result of parametric convergence of variational inequalities, which is applicable beyond the context of games.
We apply our theoretical findings to a multi-market Cournot game with transportation costs and maximum market capacity.

\end{abstract}
\section{Introduction}

\lettrine{A}VERAGE aggregative games are used to describe populations of non-cooperative agents where each agent
is not subject to one-to-one interactions with the others,
but is rather influenced by the average strategy of the entire population.
These games
are gaining more and more attention in the control community for their ability to model a vast number of technological applications
ranging from traffic~\cite{correa2004selfish} and wireless systems~\cite{alpcan2002cdma} to electricity~\cite{ma2013decentralized} and commodity markets~\cite{johari2004efficiency}.

Applying game theoretical concepts to engineering systems becomes challenging
in presence of a large number of agents with private cost functions and constraints,
willing to exchange information only with a (small) subset of the population.
Moreover, often the agents' strategies must collectively satisfy some physical coupling constraints,
as in electricity markets~\cite{kar2012distributed}, where the overall energy demand should not exceed the grid capacity,
or in communication networks~\cite{pan2009games}, where the package traffic should not exceed the congestion level.




\subsection*{Main contributions} 
To overcome the difficulties described above, in this paper we present a \textit{distributed algorithm} that guarantees convergence to an almost Nash equilibrium of a population average aggregative game \textit{with affine coupling constraints} by using only local communications over a \textit{sparse network}.
Our method works for populations of heterogeneous agents with generic local convex constraints and
generic convex cost functions
(i.e. we do not need to assume that the cost functions are quadratic,
as in \cite{grammatico:parise:colombino:lygeros:14}).
To prove algorithmic convergence we rely on two key steps:
i) we show that the Nash equilibrium of the average aggregative game of interest can be approximated to any desired precision by the Nash equilibrium of an auxiliary network aggregative game (as defined in~\cite{parise2015network});
ii) we propose a distributed algorithm to compute the Nash equilibrium of the auxiliary game.

As a side contribution, to tackle i) we prove a result on convergence of parametric variational inequalities (VIs) which is  applicable beyond the context of games.
A detailed explanation of how our result improves the existing literature on convergence of parametric VIs is given in Appendix~B.  Moreover, to tackle ii), we derive a distributed  algorithm to find a Nash equilibrium of any
network aggregative game.
Both these results can find application beyond the specific context of this paper and are thus of interest on their own. 

To illustrate our theoretical findings,  we study in detail a Cournot game with transportation costs, as introduced in~\cite[Section 1.4.3]{facchinei2007finite}.
We note that such setup extends the Cournot game~\cite[Section 2]{jensen2010aggregative} and
the multi-market Cournot game~\cite[Section 7.1]{yi2017distributed} by introducing transportation costs.
Our novel contributions consist in
introducing coupling constraints,
performing a specific study of the properties of a fundamental operator associated with the game,
and focusing on  distributed  convergence.

\subsection*{Comparison with the literature} 

The  problem of coordinating the agents to a Nash equilibrium is  of great importance for  control purposes and  it has consequently been addressed by many authors in the last few years. In the following we review such rapidly growing literature by distinguishing whether the proposed algorithms
1) can only be applied when the agent strategy sets are decoupled or allow for  constraints coupling the agents' strategies 
and 2) require  a central coordinator (decentralized algorithms) or employ only local communications (distributed algorithms).
We note that in some applications distributed algorithms are preferable to decentralized ones,
for reasons of privacy, security or for the lack of a central operator.

A vast literature  focused on the case of \textit{decoupled} strategy sets,
where the feasible strategy set of each agent is not affected by the strategies of the other agents.
Decentralized algorithms relying on the presence of a central operator
that can gather and broadcast information to the agents
are suggested in~\cite{grammatico:parise:colombino:lygeros:14,dario2015aggregative}.
On the other hand, distributed algorithms relying only on local communications among the agents
are suggested in~\cite{koshal2012gossip,chen2014autonomous,parise2015networkA,koshal2016distributed}.

All the previously mentioned algorithms cannot be applied  to the case of  \textit{coupling constraints}, 
because they build on the core assumption that the strategy sets are decoupled.
Consequently, including coupling constraints requires a fundamental rethinking of the algorithms suggested in the literature,
as we highlight in detail in Section~\ref{sec:literature_relation}.
A decentralized algorithm overcoming this issue for average aggregative games has been recently suggested in~\cite{paccagnan2016distributed,gentile2017nash},
by exploiting a primal-dual reformulation of the Nash equilibrium problem originally introduced for generic games
in~\cite{rosen1965existence} and~\cite[Theorem 3.1]{facchinei2007generalized_2}.
To the best of our knowledge,
the only distributed algorithm available in the literature for average aggregative games with coupling constraints is~\cite{liang2016distributed}.
However such algorithm is only applicable if the coupling constraints can be expressed as
the solution set of a system of linear equations~\cite[eq. (5)]{liang2016distributed}.
This is a very restrictive assumption that prevents the applicability of the algorithm suggested in \cite{liang2016distributed} to any of the practical cases discussed before.

We finally note that our work has some affinity with the distributed algorithms suggested in \cite{yi2017distributed,frazzoli,yin2011nash} to compute a Nash equilibrium of generic games (i.e., games that do not have the aggregative structure considered here) with coupling constraints.
The term ``distributed'' in all these references, however,  refers to the fact that any specific agent is only allowed to communicate with the agents that  affect his  cost function.
In average aggregative games the cost function of each agent is affected by the strategy of all the other agents, because it is affected by the average population strategy.
Consequently, the schemes proposed in~\cite{frazzoli,yin2011nash,yi2017distributed} can theoretically  be applied to average aggregative games, but they would require communications among all the agents.

The average aggregative game (AAG) considered in this paper also relates to the class of mean field games (MFGs),
since both in AAGs and in MFGs the agents are influenced only by an aggregate of
the population strategies~\cite{huang2007large,lasry2007mean,bensoussan2016linear,huang2012social,bensoussan2013mean}.
There are however some fundamental differences that make results derived for MFGs not applicable in our context. Firstly, in MFGs the strategies of the agents are unconstrained: the coupling constraints that we consider here cannot be handled with the mean field approach based on partial differential equations. Secondly, MFG results are derived for the limit of infinite population size: we here instead consider populations of any size. Thirdly, in MFGs agents are typically homogeneous or have prior (probabilistic) information of the parameters of the other agents: we instead assume that the agents have no information on the rest of the population and rely on local communications over a network. Finally, MFGs are stochastic dynamic games while AAGs are deterministic static games.

\textit{Organization:}
In Section~\ref{sec:agg_games} we formulate the game setup, we present our main algorithm and we state our main result.
In Section \ref{sec:auxiliary_results} we present some preliminary results, needed to prove our main theorem.
The main result, stated in Section~\ref{sec:agg_games}, is proven
in Sections~\ref{sec:relation} and~\ref{sec:iterative_scheme}.
Specifically, in Section~\ref{sec:relation} we study the relation between the Nash equilibrium of the auxiliary game (which is parametrized by $\nu\in\mathbb{N}$) and the Nash equilibrium of the original average aggregative game, which can be seen as the limiting game when $\nu\rightarrow \infty$. In Section~\ref{sec:iterative_scheme} we prove convergence of our algorithm to the Nash equilibrium of the auxiliary game.
Section~\ref{sec:application} focuses on the application, while Section~\ref{sec:conclusion} presents some  some generalizations of the previous theory, which were omitted to keep the exposition simple, and  future research directions. 
Appendix B is a standalone section containing the result on convergence of parametric VIs.
Appendices A and C contain some auxiliary definitions and results, respectively.
We report our main results as theorems, our auxiliary results as lemmas and auxiliary results that already exist in the literature as propositions.

\textit{Notation:} 
\fontdimen2\font=3pt
$||x||$ is the 2-norm of $x \in \R^n$.
$I_n \in \R^{n \times n}$ is the identity matrix, $\ones[n] \in \R^{n}$ is the vector of unit entries, $\zeros[n] \in \R^{n}$ is the vector of zero entries,
$e_i\in\R^n$ is the $i^\text{th}$ canonical basis vector.
Given $A\in\mathbb{R}^{n\times n}$, $A\succ0$ ($\succeq0$) $\Leftrightarrow$ $x^\top A x>0~(\ge0),$ $\forall x\neq 0$;
$\| A \|$ is the induced 2-norm of $A$.
$\textup{diag}(A)$ is the diagonal matrix which has the same diagonal of $A$.
$\textup{blkdiag}(A_1,\dots,A_N)$ is the block diagonal matrix whose blocks are the matrices $A_1, \dots, A_N$.
Given $\N$ vectors each in $\R^{n}$, $[x^1;\ldots;x^\N] \defeq [x^i]_{i=1}^\N\defeq[{x^1}^\top,\ldots ,{x^\N}^\top]^\top \in \R^{\N n}$ and $x^{-i}\coloneqq[x_1;\dots;x_{i-1};x_{i+1};\dots;x_\N]\in \R^{(\N-1)n}$.
Given $g(x):\mathbb{R}^n \rightarrow \mathbb{R}^m$ we define $\nabla_x g(x) \in \mathbb{R}^{n\times m}$ with
$[\nabla_x g(x)]_{i,j}\coloneqq \frac{\partial g_j(x)}{\partial x\i}$.
Given $g(x):\mathbb{R} \rightarrow \mathbb{R}$, we denote $g'(x) \defeq \frac{\partial g(x)}{\partial x}$.
Given sets $\mathcal{X}^1,\dots, \mathcal{X}^{\N} \subseteq \R^n$,
we denote $\frac{1}{\N} \sum_{i=1}^\N \mc{X}^i \defeq \{z \in \R^n \vert z = \frac{1}{\N} \sum_{i=1}^{\N} x\i, \text{for some } x\i \! \in \! \mc{X}\i \}$,
the convex hull of $\mc{X}^1,\dots,\mc{X}^N$ as $\text{conv}(\mc{X}^1,\dots,\mc{X}^N)$
and $\X^{-i} \defeq \mathcal{X}^1\times\ldots\mathcal{X}^{i-1}\times\mathcal{X}^{i+1}\times\ldots\mathcal{X}^N$.
$\Pi_{\mathcal{X}}[x]$ is the projection of the vector $x$ onto the set $\mathcal{X}$.
For $a, b \in \Z$, $a \leq b$, $\Z[a,b] \defeq [a,b] \cap \Z$.

\section{Problem formulation and main result}
\label{sec:agg_games}

\subsection{Average aggregative games}
Consider a population of $N \in \mathbb{N}$ agents,
where agent $i \in \Z[1,N]$ chooses his decision variable $x^i$ in his individual constraint set $\mathcal{X}^i \subseteq \R^n$,
and interacts with the other agents via the average of their strategies.
The aim of agent $i$ is to minimize his cost function
$ \textstyle J^i\left(x^i, \sigma_\infty(x) \right)=\textstyle J^i\left(z_1, z_2\right)\mid_{z_1=x^i,z_2=\sigma_\infty(x)}$, where $J^i\left(z_1, z_2 \right) : \mc{X}^i \times \text{conv}(\mc{X}^1,\dots,\mc{X}^N) \rightarrow \R$
and\footnote{The subscript $\infty$ does not refer to an infinite population, but to the fact that the agents interact through the \textit{exact} average $\frac 1N\sum_{j=1}^N x^j$, see also the organization section.}
\begin{equation*}
\sigma_\infty(x) \defeq \frac 1N\sum_{j=1}^N x^j.
\end{equation*}
%
We denote by $x \defeq [x^1;\ldots;x^\N] \in \X \defeq \mc{X}^1 \times \dots \times \mc{X}^N \subset \R^{Nn}$ and
we assume that, besides the individual constraints,
the agents have to satisfy a linear coupling constraint expressed on the average strategy
\begin{equation}
x \in \mc{C}_\infty  \defeq \{x \in \R^{\N n}\,\vert\, \hat A \sigma_\infty(x) \le \hat b \},
\label{eq:coupling_constraints}
\end{equation}
with $\hat A \in\R^{m\times{n}}$, $\hat b \in \R^{m}$, for some $m>0$. 
The coupling constraints in~\eqref{eq:coupling_constraints} can model the fact that the usage level for a certain commodity cannot exceed a fixed capacity,
as in~\cite{kar2012distributed} and in~\cite[Fig. 4]{paccagnan2016distributed}.
The strong modeling flexibility of linear coupling constraints is further discussed in~\cite[Remark 3.1]{yi2017distributed}.

The cost and constraints just introduced give rise to the average aggregative game (AAG)
\begin{equation}
 \mc{G}_\infty \defeq \left\{ 
 \begin{aligned}
&\! \textup{ agents}:  \; && (1,\dots,N) \\
&\! \textup{ cost of agent } i:\quad &&J^i(x^i,\sigma_\infty(x)) \\
&\! \textup{ individual constraint} : &&\mc{X}^i\\
&\! \textup{ coupling constraint} :  &&\mc{C}_\infty.
\end{aligned}\right. 
\label{eq:GNEP}
\end{equation}
%
The following conditions on cost functions and constraints of $\G_\infty$  are assumed to hold throughout the rest of the paper.
\begin{sassumption*}
For each agent $i,$
the individual constraint set $\mathcal{X}^i \subset \R^n$ is convex, compact and has non-empty interior. 
The cost function $J^i(x^i,\sigma_\infty(x))$  is convex in $x^i$ for all $x^{-i}
\in \X^{-i}$
and $J^i(z_1,z_2)$ is continuously differentiable  in $z_1,z_2$. \hfill $\square$
\end{sassumption*}

Let us denote
$\mc{Q}_\infty \defeq \mc{X}\cap\mc{C}_\infty $,
and 
\begin{equation*}
\begin{aligned}
\mc{Q}_\infty ^i(x^{-i}) &\coloneqq \{x^i\in\mc{X}^i\, \vert \, x\in\mathcal{C}_\infty \}\\
&= \{x^i\in\mc{X}^i\, \vert \, \hat A \sigma_\infty(x) \le \hat b\}.
\end{aligned}
\end{equation*}
\noindent The well-known concept of Nash equilibrium for games with coupling constraints~\cite{facchinei2007generalized} can be specialized to $\Ginf$ as follows.
%

\begin{definition}[Nash Equilibrium]\label{def:NE}
A set of strategies $\bar x = [\bar x^1; \dots; \bar x^N] \in\mc{Q}_\infty$
is an $\varepsilon$-Nash equilibrium of $\Ginf$
if for all $ i\in \Z[1,N] $ and all $ x^i \! \in \! \mc{Q}_\infty^i(\bar x^{-i})$ 
\begin{align}
 J^i(\bar x^i,\sigma_\infty(\bar x)) \! \le \textstyle \! J^i\left( \! x\i, \frac1N x\i \! + \! \sum_{j \neq i} \frac1N \bar x^j \right) + \varepsilon\,. 
\label{eq:def_GNE}
\end{align}
If~\eqref{eq:def_GNE} holds with $\varepsilon = 0$ then $\bar x$ is a Nash equilibrium. 
\hfill $\square$
\end{definition}

In words, a feasible set of strategies $\left\{ \bar x^i \right\}_{i=1}^\N$ is a Nash equilibrium if
no agent can improve his cost by changing his strategy,
if the strategies of the other agents are fixed.
A Nash equilibrium for a game with coupling constraints is usually called a generalized Nash equilibrium to denote the fact that the set of feasible strategies $\mathcal{Q}^i_\infty$ for each agent $i$ depends on the strategies $x^{-i}$ of the other agents~\cite{facchinei2007generalized};
in the following we omit the word generalized for brevity.

\subsection{Communication limitations}

Our main objective is to coordinate the agents' strategies to a Nash equilibrium by using
a distributed algorithm that only requires communications over a pre-specified (sparse) communication network.
We model such network by its adjacency matrix $T \in [0,1]^{N \times N}$,
where the element $T_{ij}\in\left[0,1\right]$ is the
weight that agent $i$ assigns to communications received from agent $j$,
with $T_{ij}=0$ representing no communication.
For brevity, we refer to $T$ as communication network, even though it is the adjacency matrix of the communication network.
Agent $j$ is an in-neighbor of $i$ if $T_{ij}>0$ and an out-neighbor if $T_{ji}>0$.
We denote the sets of in- and out-neighbors of agent $i$ as $\mathcal{N}^i_{\textup{in}}$ and $\mathcal{N}^i_{\textup{out}}$, respectively.

We introduce the following assumption on the communication network $T$.
\begin{assumption}[Communication network]
\label{ass:primitive_doubly}
The communication matrix $T$  is primitive and doubly stochastic.
\hfill $\square$
\end{assumption}
The definitions of primitive and doubly-stochastic can be found in~\cite{olfati2007consensus},
where  graph theoretical conditions guaranteeing  Assumption~\ref{ass:primitive_doubly} are also presented.
Loosely speaking, Assumption~\ref{ass:primitive_doubly} ensures that if the agents communicate a sufficiently large number of times over $T$,
they are able to recover the average of the strategies across the entire population.

\subsection{Main result: A distributed algorithm for $\Ginf$}
To compute an almost Nash equilibrium in a distributed fashion,
we propose the following Algorithm~\ref{alg:NAG_nu},
where at iteration $k$ each agent $i$ updates four variables:
\begin{itemize}
\item[-] his strategy $x^i_{(k)}$,
\item[-] a dual variable $\lambda^i_{(k)}$ relative to the coupling constraint $\mathcal{C}_\infty$,
\item[-] a local  average of his in-neighbors' strategies $\sigma^i_{\nu,(k)}$,
\item[-] a local average of his out-neighbors' dual variables $\mu^i_{\nu,(k)}$.
\end{itemize}
To overcome the fact that the communication network is sparse we assume that to compute $\sigma^i_{\nu,(k)}$ and $\mu^i_{\nu,(k)}$ the agents  communicate  not once but multiple times over the network $T$.
The  number of communications per update is denoted by $\nu\in\mathbb{N}$ and is a  tuning parameter of the algorithm. 
The personal strategy (or primal variable) and the dual variable, in turn, are updated by a gradient-like step that depends on a second tuning parameter $\tau>0$.
In particular, the strategy update step is similar to that of the standard projection algorithm~\cite[Algorithm 12.1.1]{facchinei2007finite}.
We finally note that both tuning parameters $\nu$ and $\tau$ are  decided a priori and do not change during the algorithm execution.

\begin{algorithm}
\caption{Distributed algorithm for $\Ginf$}
\label{alg:NAG_nu}

\textbf{Initialize}: 
Agent $i$ with state $x^i_{(0)}\in \X\i$ and dual variable $\lambda_{(0)}^i \! \in \! \R^m_{\ge0}$.
Set $\tau>0$, $\nu \in \mathbb{N}$,  $k =0$, $\sigma^i_{\nu,(0)} =x^i_{(0)}$.
\vspace{0.1cm}

\textbf{Iterate until convergence}:
\vspace*{-0.1cm}
{\small
\begin{subequations}
\label{eq:apa_inner}
\begin{align*}
& \hspace{-0.2cm}\mbox{{\textit{Communication: Dual}}}\\
&\left\lfloor \begin{array}{l}
\mu^i_{\nu,(k)}  \textstyle \leftarrow \lambda^{i }_{(k)}, \forall \, i \\
\mbox{repeat } \nu  \mbox{ times}\\
\qquad  \mu^i_{\nu,(k)}  \textstyle \leftarrow \sum_{j \in \mc{N}^i_\text{out}} T_{ji} \, \mu^j_{\nu,(k)}, \forall \, i
\hspace*{1.7cm} \\ 
\end{array}	 \right. \\[-0cm]
& \hspace{-0.2cm}\mbox{{\textit{Update: Primal }}}\\
&\left\lfloor \begin{array}{l}
F^i_{\nu,(k)}\leftarrow  \nabla_{\!\! z_1} \!  J^i(x^i_{(k)},\sigma^i_{\nu,(k)}\! ) + [T^\nu]_{ii} \nabla_{\!\! z_2}
J^i(x^i_{(k)},\sigma^i_{\nu,(k)}), \forall \, i \\
x^{i}_{(k+1)}\leftarrow \! \Pi_{\mathcal{X}^i}[x^{i }_{(k)} \! - \! \tau ( F^i_{\nu,(k)} + \! {\hat A}^\top \! \mu^i_{\nu,(k)}) ], \forall \, i \\
\end{array}	 \right. \\
& \hspace{-0.2cm}\mbox{{\textit{Communication: Primal}}}\\
&\left\lfloor \begin{array}{l}
\sigma^i_{\nu,(k+1)} \textstyle \leftarrow x^{i }_{(k+1)}, \forall \, i  \\
\mbox{repeat } \nu  \mbox{ times}\\
\qquad \sigma^i_{\nu,(k+1)}  \textstyle \! \leftarrow \sum_{j \in \mc{N}^i_\text{in}} T_{ij} \sigma^j_{\nu,(k+1)}, \forall \, i \hspace*{-0.3cm}\\
\end{array}	 \right. \\
& \hspace{-0.2cm}\mbox{{\textit{Update: Dual}}}\\
&\left\lfloor \begin{array}{l}
\lambda^i_{(k+1)}  \leftarrow \! \Pi_{\mathbb{R}^{m}_{\ge0}}[\lambda^i_{(k)}-\tau (\hat b-2\hat A\sigma^i_{\nu,(k+1)} +\hat A\sigma^i_{\nu,(k)})], \forall \, i \label{eq:apa_inner_c}\end{array}	 \right. \\
& \hspace{-0.15cm} k \leftarrow k+1 \\[-0.9cm]
\end{align*}
\end{subequations}}
\end{algorithm}
  
The main objective of the paper is to prove convergence of Algorithm \ref{alg:NAG_nu} to an $\varepsilon_\nu$-Nash equilibrium of  $\mathcal{G}_\infty$,
where $\varepsilon_\nu\rightarrow 0$ as $\nu\rightarrow \infty$.
To this end, we  make use of the following additional assumptions.

\begin{assumption}[Coupling constraints]\label{ass:R2}
The matrix $\hat A$ and the vector $\hat b$ in \eqref{eq:coupling_constraints} are such that the following implication holds.
\begin{equation*}
\{ \hat A^\top  \hat s= 0 ,\quad   \hat b^\top  \hat s\le 0    , \quad  \hat s\ge 0  \} \quad \Rightarrow  \quad  \hat s=0.
\tag*{$\square$}
\end{equation*}
\end{assumption}

It can be shown that if the coupling constraint~\eqref{eq:coupling_constraints} consists of an upper and a lower bound for $\sinf(x)$,
i.e. it is of the form $b_1 \le \sinf(x) \le b_2$, with $b_1 < b_2$, then Assumption~\ref{ass:R2} is satisfied.
We show in the application Section~\ref{sec:application}
another example of coupling constraint satisfying Assumption~\ref{ass:R2}.

\begin{assumption}[Regularity of cost functions]
\label{ass:F_Fnu_SMON}
The operator $\Finf: \, \X \to \R^{Nn}$, defined as
\begin{equation}
x \mapsto F_\infty(x) \defeq [\nabla_{x^i} J^i(x^i, \sinf(x) )]_{i=1}^N
\label{eq:Finf}
\end{equation}
is  strongly monotone.
Further, the cost function $J\i(z_1,z_2)$ is twice continuously differentiable for all $i \in \Z[1,N]$.
\hfill $\square$
\end{assumption}
\noindent We recall in the following the definition of strong monotonicity.

\begin{definition}[Strong monotonicity]
\label{def:SMON}
An operator $H: \K \subseteq \R^n \to \R^n$ is said to be strongly monotone if there exists $\alpha_H > 0$ such that
for all $x,y \in \K$
\begin{equation*}
(H(x)-H(y))^\top(x-y) \ge \alpha_H \| x - y \|^2.
\tag*{$\square$}
\end{equation*}
\end{definition}

We note that sufficient conditions for Assumption \ref{ass:F_Fnu_SMON} to hold
have been discussed in~\cite[Lemmas 3 and 4, Corollaries 1 and 2]{gentile2017nash}
for specific instances of average aggregative games.

Our main result is stated in the following theorem.

\begin{theorem}
\label{thm:main}
If Assumptions~\ref{ass:primitive_doubly},~\ref{ass:R2} and~\ref{ass:F_Fnu_SMON} hold, then
for every precision $\varepsilon > 0$, there exists a minimum number of communications $\nueps > 0$ such that, for every $\nu > \nueps$ and for every initial condition $(x_{(0)},\lambda_{(0)}) \in \X \times \R^{Nm}_{\ge 0}$,
the sequence $(x_{(k)})_{k=1}^\infty$ produced by Algorithm~\ref{alg:NAG_nu} converges to an $\varepsilon$-Nash equilibrium of $\G_\infty$ for $\tau$ small enough. \hfill $\square$
\end{theorem}
A precise bound on $\tau$ is given in Theorem \ref{thm:convergence_alg} in Section \ref{sec:iterative_scheme}.

\section{Auxiliary results}
\label{sec:auxiliary_results}
In this section we present some preliminary results needed for the proof of Theorem \ref{thm:main},
whose main idea consists in defining an auxiliary game $\mathcal{G}_\nu$ parametric in the number of communications $\nu$ and showing that
\begin{enumerate}
\item Algorithm~\ref{alg:NAG_nu}  converges to a specific Nash equilibrium of $\mathcal{G}_\nu$, called \textit{variational equilibrium};
\item the variational Nash equilibrium of $\mathcal{G}_\nu$ is an $\varepsilon_\nu$-Nash equilibrium of $\mathcal{G}_\infty$, with $\varepsilon_\nu\rightarrow 0$ as $\nu\rightarrow \infty$.
\end{enumerate}
To prep the field for these two results, which are proven in Section~\ref{sec:proof},
we here define the game $\mathcal{G}_\nu$,
introduce some basic results from the theory of variational inequalities (VI)
and study the relation between the VI operators associated with the two games $\mathcal{G}_\infty$ and $\mathcal{G}_\nu$.

\subsection{Multi-communication network aggregative games}

In each iteration of Algorithm~\ref{alg:NAG_nu} the agents need to communicate $\nu$ times over $T$;
mathematically this is equivalent to communicating once over a fictitious network with
adjacency matrix $T^\nu$.
Based on $T^\nu$, we introduce the local averages
\begin{equation*}
\sigma^i_\nu(x) \defeq \sum_{j=1}^N [T^\nu]_{ij} x\j.
\end{equation*}
We define $\Gnu$ as a game with same constraints and cost functions as in $\G_\infty$
except for the fact that each agent reacts to the local average $\sigma^i_\nu(x)$ instead of the global average $\sigma_\infty(x)$.
Specifically, upon defining
\begin{equation*}
\mc{C}_\nu \defeq \{x \in \R^{\N n}\,\vert\, \hat A \sigma^j_\nu(x) \le \hat b, \, \forall j \in \Z[1,N] \}
\end{equation*}
we formally introduce the multi-communication network aggregative game as
\begin{equation*}
 \mc{G}_\nu \defeq \left\{ 
 \begin{aligned}
&\! \textup{ agents}:  \; && (1,\dots,N) \\
&\! \textup{ cost of agent } i:\quad &&J^i(x^i,\sigma\i_\nu(x)) \\
&\! \textup{ individual constraint} : &&\mc{X}^i\\
&\! \textup{ coupling constraint} :  &&\mc{C}_\nu.
\end{aligned}\right. 
\end{equation*}
The definition of Nash equilibrium for $\Gnu$ is the analogous of Definition~\ref{def:NE} for $\Ginf$.
To analyze the relation between $\mathcal{G}_\infty$ and $\mathcal{G}_\nu$ we use the framework of variational inequalities introduced next.

\subsection{Basics of variational inequalities}
A fundamental fact used throughout the rest of the paper is that a specific class of
Nash equilibria of any convex game,
called variational Nash equilibria,
can be obtained by solving a variational inequality constructed from the game primitives.

\begin{definition}[Variational inequality]
Given a set $\K \subseteq \R^n$ and an operator $H: \K \to \R^n$,
the point $\bar x \in \R^n$ is a solution of $\textup{VI}(\K,H)$ if it satisfies
\begin{align}
H(\bar x)^\top (x-\bar x) \ge 0, \;\; \forall \, x \in \K.
\label{eq:VI_def}
\tag*{$\square$}
\end{align}
\end{definition}
\noindent A discussion on how variational inequalities generalize convex optimization programs can be found in~\cite[Section 1.3.1]{facchinei2007finite}.
In the following we report a sufficient condition for existence and uniqueness of the solution of a variational inequality.

\begin{proposition}[{\cite[Theorem 2.3.3.b]{facchinei2007finite}}]
\label{prop:uniqueness}
Consider a closed and convex set $\K$ and a strongly monotone operator $H: \K \to \R^n$.
Then VI($\K$,$H$) admits a unique solution.
\hfill $\square$
\end{proposition}

The following lemma gives a more intuitive characterization of the strong monotonicity property.
\begin{proposition}[\textup{\cite[Proposition 2.3.2]{facchinei2007finite}}]
\label{prop:pd}
A continuously differentiable operator $H: \mc{K} \subseteq \R^\di \to \R^\di$ is strongly monotone with monotonicity constant $\alpha$ if and only if $\nabla_x H(x)\succeq \alpha I_n$ for all $x \in \mc{K}$. \hfill{$\square$}
\end{proposition}

\subsection{Variational Nash equilibria}

To draw the connection between VIs and Nash equilibria, let us introduce the following quantities relative to $\Gnu$
\begin{subequations}
\label{eq:game_nu_nag}
\begin{align}
\Fnu(x) &\defeq [\nabla_{x^i} J^i(x^i, \snu^i(x) )]_{i=1}^N \label{eq:game_nu_nag:F},\\
\Qnu &\defeq \{ x \in \X \vert \Anu x \le \binf \}, \label{eq:game_nu_nag:Q} \\
\Qnu^i(x^{-i}) &\defeq \{ x\i \in \X^i \vert \Anu x \le \binf\}, \label{eq:game_nu_nag:Qi} \\
\Anu &\defeq T^\nu \otimes \hat A \label{eq:game_nu_nag:A},\\
b &\defeq \ones[N] \otimes \hat b \label{eq:game_aag:b}
\end{align}
\end{subequations}
and recall the corresponding quantities relative to $\Ginf$
\begin{subequations}
\label{eq:game_aag}
\begin{align}
\Finf(x) &\defeq [\nabla_{x^i} J^i(x^i, \sinf(x) )]_{i=1}^N,  \label{eq:game_aag:F}\\
\Qinf &\defeq \{ x \in \X \vert \Ainf x \le \binf\}, \label{eq:game_aag:Q} \\
\Qinf^i(x^{-i}) &\defeq \{ x\i \in \X^i \vert \Ainf x \le \binf\}, \label{eq:game_aag:Qi} \\
\Ainf &\textstyle \defeq \frac{1}{N} {\ones[N]}{\ones[N]}^\top \otimes \hat A.  \label{eq:game_aag:A}
\end{align}
\end{subequations}
The operator in~\eqref{eq:game_aag:F} is the same as~\eqref{eq:Finf} in Assumption~\ref{ass:F_Fnu_SMON}
and in~\eqref{eq:game_aag} the coupling constraint $\mathcal{C}_\infty$ is expressed in the redundant form $\Ainf x \le b$
(consisting of $N$ repetitions of the constraint
$\hat A \sinf(x) \le \hat b$) to match the structure of $\Anu x \le b$ in~\eqref{eq:game_nu_nag}.

In the following we specialize a well-known result of~\cite[Theorem 2.1]{facchinei2007generalized_2} to the two games $\G_\infty$ and $\G_\nu$.
\begin{proposition}[Variational Nash equilibrium {\cite[Theorem 2.1]{facchinei2007generalized_2}}]
\label{prop:var_GNE_is_GNE}
Every solution $\bar x_\infty$ of VI($\Qinf$,$\Finf$) is a Nash equilibrium of $\G_\infty$,
called a variational Nash equilibrium of $\G_\infty$.
Moreover, if $J\i(x\i,\sigma\i_\nu(x))$ is convex in $x\i$ for all $x^{-i}\in\X^{-i}$,
every solution $\bar x_\nu$ of VI($\Qnu$,$\Fnu$) is a Nash equilibrium of $\G_\nu$,
called a variational Nash equilibrium of $\G_\nu$.
\hfill $\square$
\end{proposition}

Due to the presence of the coupling constraints the converse of Proposition~\ref{prop:var_GNE_is_GNE} does not hold.
In other words, there can be Nash equilibria of $\Ginf$ (resp. $\Gnu$) that
cannot be obtained as solutions of VI($\Qinf$,$\Finf$) (resp. VI($\Qnu$,$\Fnu$)).
Within the class of Nash equilibria, variational Nash equilibria enjoy special stability and sensitivity properties~\cite{facchinei2007generalized}.
Variational equilibria are also a subset of the \textit{normalized equilibria} defined in~\cite{rosen1965existence},
which are in most cases the only Nash equilibria that can be computed.
Variational equilibria can be interpreted as the most fair within the class of normalized equilibria,
because the burden of meeting the coupling constraints is divided equally among all the agents~\cite[Theorem 3.1]{facchinei2007generalized_2}.

We conclude by stating convergence of the operator $F_\nu$ to $F_\infty$ as $\nu$ tends to infinity.

\begin{lemma}
\label{lemma:lim_A_nu_F_nu}
Under Assumption \ref{ass:primitive_doubly}, 
\begin{align*}
\lim_{\nu\rightarrow \infty} T^\nu = \! \frac{1}{N}{\ones[N]}{\ones[N]}^\top \!\!, \;
\qquad \lim_{\nu \to \infty} A_\nu = A_\aag, \;
\end{align*}
and
\begin{equation*}
\lim_{\nu \to \infty} F_\nu(x) = F_\aag(x)
\end{equation*}
uniformly in $x$.
The operators $F_\nu$ are bounded on $\mathcal{X}$, uniformly in $\nu$. \hfill $\square$
\end{lemma}
\begin{proof}
The fact that $\lim_{\nu\rightarrow \infty} T^\nu = \frac{1}{N}{\ones[N]}{\ones[N]}^\top$ is proven in~\cite[Theorem 2]{olfati2007consensus}.
Convergence of $A_\nu$ to  $A_\aag$ follows immediately from the definitions~\eqref{eq:game_nu_nag:A},~\eqref{eq:game_aag:A} and the properties of the Kronecker product. 
Note that
\begin{align*}
F_\aag(x)&=\textstyle [\nabla_{z_1} J^i(x^i,\sigma_\aag(x))+\hspace{0.1cm} \frac1N \hspace{0.1cm}\nabla_{z_2} J^i(x^i,\sigma_\aag(x))]_{i=1}^N,\\
F_\nu(x)&=\textstyle [\nabla_{z_1} J^i(x^i,\sigma^i_\nu(x))+[T^\nu]_{ii}\nabla_{z_2} J^i(x^i,\sigma^i_\nu(x))]_{i=1}^N.
\end{align*}
Uniform convergence of $F_\nu$ to $F_\aag$ follows by continuity of
$\nabla_{z_1} J^i(z_1,z_2)$ and $\nabla_{z_2} J^i(z_1,z_2)$ in $z_1,z_2$ for all $i$ (ensured by the Standing Assumption),
by $[T^\nu]_{ii}\rightarrow \frac 1N$, since $T^\nu\rightarrow \frac 1N \ones[N]\ones[N]^\top$,
and by $\sigma^i_\nu(x) \to \sigma_\aag(x)$ uniformly in $x$.

Finally, as $\|\nabla_{z_1} J^i(z_1,z_2)\|$ and $\|\nabla_{z_2} J^i(z_1,z_2)\|$
are continuous functions over the compact set $\mathcal{X}^i\times\textup{conv}\{\mathcal{X}^1,\ldots,\mathcal{X}^N\}$,
there exists $M > 0$ such that $\|\nabla_{z_1} J^i(z_1,z_2)\| < M$ and $\|\nabla_{z_2} J^i(z_1,z_2)\| < M$.
Note that $[T^\nu]_{ii} \le 1$ for all $i\in\Z[1,N]$ and  for all  $\nu>0$, since $T$ and thus $T^\nu$ are non-negative and doubly stochastic.
Then for all $x\in\mathcal{X}$
\begin{align*}
\|F_\nu(x)\|^2&= \sum_{i=1}^N \! \|\nabla_{\! z_1} \! J^i(x^i \!,\sigma^i_\nu(x))\!+\![T^\nu]_{ii} \nabla_{\! z_2} J^i(x^i \!,\sigma^i_\nu(x))\|^2  \\[-0.4cm]
&\le  \sum_{i=1}^N( M^2+ M^2 +2M^2) = 4N M^2,
\end{align*}
which proves that $F_\nu$ is bounded, uniformly in $\nu$. 
\end{proof}

%
%

\section{Proof of Theorem \ref{thm:main}}
\label{sec:proof}

\noindent We divide the proof of Theorem \ref{thm:main} into two parts.
\begin{enumerate}
\item
 In subsection~\ref{sec:relation} we show that
$\bar x_\nu \rightarrow \bar x_\aag$ as $\nu \to \infty$ and that $\bar x_\nu$ is an $\varepsilon_\nu$-Nash equilibrium
for $\G_\infty$, with $\varepsilon_\nu \rightarrow 0$ as $\nu\rightarrow \infty$.
To this end we
exploit a novel result on parametric convergence of variational inequalities,
which is derived in Appendix B. 
\item In subsection~\ref{sec:iterative_scheme} we show that Algorithm \ref{alg:NAG_nu}
 converges to $\bar x_\nu$.
\end{enumerate}
From these two results it follows that, for any desired $\varepsilon$ one can guarantee that the limit point $\bar x_\nu$ of Algorithm \ref{alg:NAG_nu} is an $\varepsilon$-Nash equilibrium of $\G_\infty$,
by setting a  number $\nu$ of communications large enough.
This proves Theorem \ref{thm:main}.

\subsection{Convergence of the variational Nash equilibrium of $\G_\nu$ to the variational Nash equilibrium of $\G_\infty$}
\label{sec:relation}



The next lemma provides a sufficient condition for $F_\nu$ to be strongly monotone, which is then used in Theorem~\ref{thm:convergence}.

\begin{lemma}
\label{lemma:c3}
If Assumptions~\ref{ass:primitive_doubly} and~\ref{ass:F_Fnu_SMON} hold,
there exists $\nuSMON$ such that $F_\nu$ is strongly monotone for all $\nu>\nuSMON$. \hfill $\square$
\end{lemma}
\noindent We report the proof of Lemma~\ref{lemma:c3} in Appendix C.

\begin{theorem}
\label{thm:convergence}
Suppose that Assumptions \ref{ass:primitive_doubly},~\ref{ass:R2} and~\ref{ass:F_Fnu_SMON}  hold. Then 
\begin{enumerate}
\item
The game $\Ginf$ has a unique variational Nash equilibrium $\bar x_\aag$ and for any $\nu>\nuSMON$,
$\Gnu$ has a unique variational Nash equilibrium $\bar x_\nu$. 
Moreover,
\begin{equation}
\lim_{\nu\rightarrow \infty}\bar x_\nu = \bar x_\aag.
\label{eq:xnu_tends_xinf}
\end{equation}
\item For every $\varepsilon > 0$, there exists a $\nueps > \nuSMON $ such that, for every $\nu > \nueps $, the variational Nash equilibrium $\bar x_\nu$ of $\Gnu$ 
is an  $\varepsilon$-Nash equilibrium of $\Ginf$. \hfill $\square$
\end{enumerate}
\label{convergence_x}
\end{theorem}

\begin{proof}
1)
Existence and uniqueness of $ \bar x_\aag$ and $\bar x_\nu$  solutions to VI($\Qinf$,$\Finf$) and VI($\Qnu$,$\Fnu$) respectively
is guaranteed by Proposition~\ref{prop:uniqueness},
because the operator $F_\aag$ is strongly monotone by Assumption~\ref{ass:F_Fnu_SMON}
and the operator $F_\nu$ (for $\nu>\nuSMON$) is strongly monotone by Lemma~\ref{lemma:c3}.
By Proposition~\ref{prop:pd},  strong monotonicity of $F_\nu$ implies convexity of
$J\i(x\i,\sigma\i_\nu(x))$ in $x\i$. Consequently,  Proposition~\ref{prop:var_GNE_is_GNE}  guarantees that $\bar x_\infty$  and $\bar x_\nu$  are the unique variational Nash equilibrium of $\mathcal{G}_\infty$ and $\mathcal{G}_\nu$, respectively.
To show~\eqref{eq:xnu_tends_xinf},
we use Theorem~\ref{thm:convergenceVI} in Appendix B,
which is a  general result on convergence of  parametric variational inequalities.
The theorem is based on Assumption~\ref{ass:mega_param_VI} in Appendix B.

To verify such assumption note that by Lemma~\ref{lemma:lim_A_nu_F_nu} $A_\nu \to A_\aag$ and, for each $x\in\mathcal{X}$, $F_\nu(x)\rightarrow F_\aag(x)$; $b_\nu=b=b_\infty$ for all $\nu$. Moreover, in our setup~\eqref{eq:condition_R2} reads as
%
%
\begin{equation}
\label{eq:proof_show_R2}
\{  A_\aag^\top   s= 0 ,\quad    b^\top   s\le 0    , \quad   s\ge 0  \}\quad \Rightarrow  \quad   s=0.
\end{equation}
To prove~\eqref{eq:proof_show_R2}
take $s\defeq[s^1;\ldots;s^N]\in\R^{Nn}$ such that $A_\aag^\top s= 0$, $  b^\top  s\le 0$ and $ s\ge 0$ and define $\hat s\defeq\sum_{j=1}^N s^j \in\R^n$. Then
\begin{align*}
A_\aag^\top  s= 0& \quad \Leftrightarrow \quad\textstyle  \left(\frac1N\ones[N]\ones[N]^\top \otimes \hat A^\top\right) s =0 \quad \Rightarrow \quad \hat A^\top \hat s =0,\\
b^\top  s\le 0 &\quad \Leftrightarrow \quad  (\ones[N]^\top \otimes \hat b^\top) s\le 0\hspace{1.5cm}  \Rightarrow\quad  \hat b^\top \hat s\le 0,\\
s\ge 0 &\quad \Rightarrow \quad  \hat s\ge 0.
\end{align*}
By Assumption~\ref{ass:R2}, we conclude $\hat s=0$. Since $ s\ge 0$, it  must be $ s= 0$ thus proving \eqref{eq:proof_show_R2}.


2) We divide the proof of this statement into two parts:
i) we prove that $\bar x_\nu \in \Qinf$ for any $\nu>0$, and
ii) we prove that condition~\eqref{eq:def_GNE} is satisfied. \\
i) Being $\bar x_\nu $ a Nash equilibrium for $\Gnu$,
$\bar x_\nu \in\mc{Q}_\nu$, hence $\bar x_\nu \in\mc{X}$ and $\hat A \sigma_\nu^i(\bar x_\nu) \le \hat b$ for all $i$.
By summing over all $i$ and dividing by $N$, we obtain
\begin{align}
\textstyle \hat A \left( \frac{1}{N} \sum_{i=1}^N \sigma_\nu^i(\bar x_\nu) \right)\le \hat b.
\label{eq:feas1}
\end{align}
However,
\begin{equation}
\begin{aligned}
&\textstyle\sum_{i=1}^N \sigma_\nu^i(\bar x_\nu) = \sum_{i=1}^N \sum_{j=1}^N [T^\nu]_{ij} \bar x^{j}_\nu \\
&\textstyle=\sum_{j=1}^N \left( \sum_{i=1}^N [T^\nu]_{ij} \right) \bar x^{j}_\nu= \sum_{j=1}^N \bar x^{j}_\nu = N \sigma_\aag(\bar x_\nu),
\label{eq:feas2}
\end{aligned}
\end{equation}
where the second to last equality holds because, by Assumption~\ref{ass:primitive_doubly},
$T$ is doubly stochastic and so is $T^\nu$.
By substituting~\eqref{eq:feas2} into~\eqref{eq:feas1}
we obtain $\hat A  \sigma_\aag(\bar x_\nu)\le \hat b$, thus $\bar x_\nu \in\mc{Q}_\aag$ for any $\nu$.
\\
ii) Since $\bar x_\nu$ is a Nash equilibrium for $\Gnu$,
for all $i\in\Z[1,N]$  and for all $x^i_\nu\in \mc{Q}^i_\nu(\bar x^{-i}_\nu)$ it holds
\begin{equation}
J^i(\bar x^{i}_\nu,\sum_j [T^\nu]_{ij} \bar x^{j}_\nu ) \le J^i(x^{i}_\nu,[T^\nu]_{ii} x^{i}_\nu+ \sum_{j\ne i} [T^\nu]_{ij}\bar x^{j}_\nu ).
\label{eq:Nash_nu_NAG}
\end{equation}
%
For all $i$, $J^i(z_1,z_2)$ is continuously differentiable hence Lipschitz.
Then there exists a common Lipschitz constant $L$ such that for all $i$,
all $z_1,z_1^a,z_1^b \in \X\i$, and all
$z_2,z_2^a,z_2^b \in\textup{conv}(\mathcal{X}^1,\ldots,\mathcal{X}^N)$
\begin{align*}
\|J^i(z_1^a,z_2)-J^i(z_1^b,z_2)\|&\le L\|z_1^a-z_1^b\|, \\
\|J^i(z_1,z_2^a)-J^i(z_1,z_2^b)\|&\le L\|z_2^a-z_2^b\|. 
\end{align*}
Let $D\defeq\max_{z\in \textup{conv}\{\mathcal{X}^1,\ldots,\mathcal{X}^N\}}\{\|z\|\}$ and $\delta(\nu)\defeq\|\frac{1}{N}{\ones[N]}{\ones[N]}^\top-T^\nu\|_\infty$.
Set $\varepsilon_1\defeq\varepsilon/(4LD)$. By Lemma~\ref{lemma:lim_A_nu_F_nu}, there exists $\nu_1>0$ such that for all $\nu>\nu_1$, $\delta(\nu)<\varepsilon_1$. Moreover,
\begin{equation}
\label{step1}
\begin{aligned}
&\textstyle J^i(\bar x^{i}_\nu,\sigma_\aag(\bar x_\nu) )=J^i(\bar x^{i}_\nu,\sum_j \frac{1}{N} \bar x^{j}_\nu )\\ 
&\le \textstyle J^i(\bar x^{i}_\nu,\sum_j [T^\nu]_{ij} \bar x^{j}_\nu )+L\|\sum_j (1/N-[T^\nu]_{ij}) \bar x^{j}_\nu\|\\
&\le \textstyle J^i(\bar x^{i}_\nu,\sum_j [T^\nu]_{ij} \bar x^{j}_\nu )+L\sum_j |1/N-[T^\nu]_{ij}| \|\bar x^{j}_\nu\|\\
&\le \textstyle J^i(\bar x^{i}_\nu,\sum_j [T^\nu]_{ij} \bar x^{j}_\nu )+L D \max_i{\{\sum_j |1/N-[T^\nu]_{ij}|\}}\\
&=\textstyle J^i(\bar x^{i}_\nu,\sum_j [T^\nu]_{ij} \bar x^{j}_\nu )+LD\delta(\nu) \\[-0.1cm]
&\overset{\eqref{eq:Nash_nu_NAG}}{\le}\textstyle J^i(x^{i}_\nu,[T^\nu]_{ii} x^{i}_\nu+ \sum_{j\ne i} [T^\nu]_{ij}\bar x^{j}_\nu )+LD\varepsilon_1 \\
&\le \textstyle J^i(x^{i}_\nu,\frac{1}{N} x^{i}_\nu+ \sum_{j\ne i} \frac{1}{N}\bar x^{j}_\nu )+2LD\varepsilon_1, 
\end{aligned}
\end{equation}
for all $x^i_\nu\in \mc{Q}^i_\nu(\bar x^{-i}_\nu)$, for all $\nu>\nu_1$.
The last inequality in \eqref{step1} can be proven with a chain of inequalities similar to the previous ones in  \eqref{step1}.
Condition~\eqref{step1} implies that~\eqref{eq:def_GNE} holds for all $x^i_\nu\in \mc{Q}^i_\nu(\bar x^{-i}_\nu)$,
as we used the fact that $\bar x_\nu$ is a Nash equilibrium for $\Gnu$.
Since we however want to prove that $\bar x_\nu$ is an $\varepsilon$-Nash for $\Ginf$,
we need to prove that~\eqref{eq:def_GNE} holds for all $x^i \in \Qinf^i(\bar x^{-i}_\nu)$.\\
To this end, take  any such $x^i\in \Qinf^i(\bar x^{-i}_\nu)$.
Set $\varepsilon_2\defeq\varepsilon/(2L+\frac{2L}{N})$.
Since we showed in the first statement that $\bar x_\nu \rightarrow \bar x_\aag$,
by Lemma~\ref{lemma:hausdorff_proj} in Appendix C,
there exists\footnote{As Lemma~\ref{lemma:hausdorff_proj} requires uniqueness of $\bar x_\nu$, we take $\nu_2>\nuSMON$. Note than $\nu_2$ is independent from $i$.}
$\nu_2>\nuSMON$ such that for all $\nu>\nu_2$ and all $i \in \Z[1,N]$ there exists $\tilde x^i_\nu\in\mc{Q}^i_\nu(\bar x^{-i}_\nu)$  such that $\|x^i-\tilde x^i_\nu\|\le \varepsilon_2$. 
From \eqref{step1} we know that since $\tilde x^i_\nu\in\mc{Q}^i_\nu(\bar x^{-i}_\nu)$ then 
\begin{equation}\label{step2}
\begin{aligned}
& \textstyle J^i(\bar x^{i}_\nu,\sigma_\aag(\bar x_\nu) )
\le J^i(\tilde x^{i}_\nu,\frac{1}{N} \tilde x^{i}_\nu+ \sum_{j\ne i} \frac{1}{N} \bar x^{j}_\nu )+2LD\varepsilon_1\\
&\le \textstyle J^i(x^{i},\frac{1}{N}  x^{i}+ \sum_{j\ne i} \frac{1}{N} \bar x^{j}_\nu )+(L+\frac{L}{N})\varepsilon_2+ 2LD\varepsilon_1\\
&\le \textstyle J^i(x^{i},\frac{1}{N}  x^{i}+ \sum_{j\ne i} \frac{1}{N} \bar x^{j}_\nu )+\varepsilon.
\end{aligned}
\end{equation}

Since \eqref{step2} holds for all $i\in\Z[1,N]$ and for all $x^i\in \Qinf^i(\bar x^{-i}_\nu)$ and given part~i),
we have proven that $\bar x_\nu$ is an
$\varepsilon$-Nash equilibrium for $\Ginf$,
for all $\nu > \nueps \defeq\max\{\nu_1,\nu_2\}$.
\end{proof}

\subsection{Convergence of Algorithm \ref{alg:NAG_nu} }
\label{sec:iterative_scheme}

\begin{theorem}
\label{thm:convergence_alg}
Suppose that for the value of $\nu$ used in Algorithm~\ref{alg:NAG_nu}
the operator $F_\nu$ in~\eqref{eq:game_nu_nag:F} is strongly monotone with constant $\alpha_\nu>0$ and Lipschitz with constant $L_\nu>0$. Set
\begin{equation}
\label{eq:condition_tau_APA}
\textstyle \tau <\frac{-L_\nu^2+\sqrt{L_\nu^4+4\alpha_\nu^2\| A_\nu\|^2}}{2\alpha_\nu \| A_\nu\|^2 }. 
\end{equation}

Then for every initial condition $(x_{(0)},\lambda_{(0)}) \in \X \times \R^{Nm}_{\ge 0}$
the sequence $(x_{(k)})_{k=1}^\infty$ produced by Algorithm~\ref{alg:NAG_nu}
converges to the unique variational Nash equilibrium of $\G_\nu$.
\hfill $\square$
\end{theorem}

\begin{proof}
Let us define $x_{(k)}\defeq[x^i_{(k)}]_{i=1}^N,\lambda_{(k)}\defeq[\lambda^i_{(k)}]_{i=1}^N,\sigma_{\nu,(k)}\defeq[\sigma^i_{\nu,(k)}]_{i=1}^N,\mu_{\nu,(k)}\defeq[\mu^i_{\nu,(k)}]_{i=1}^N$. Then the communication steps are equivalent to 
\begin{align*}
& \sigma_{\nu,(k)} \leftarrow \textstyle (T^\nu \otimes I_n) \ x_{(k)}, \\
& \mu_{\nu,(k)} \leftarrow \textstyle (T^\nu \otimes I_m)^\top \lambda_{(k)}.
\end{align*}
Consequently, the update steps can be rewritten as
\begin{align*}
x^{i}_{(k+1)} &\leftarrow \Pi_{\mathcal{X}^i}[x^{i }_{(k)}- \tau (F^i_{\nu,(k)}  + \textstyle \! {\hat A}^\top \!\! \sum_{j=1}^N [T^\nu]_{ji} \lambda^j_{(k)} )], \! \\
 \lambda^i_{(k+1)} &\leftarrow \! \Pi_{\mathbb{R}^{m}_{\ge0}}[\lambda^i_{(k)}-\tau \cdot\\
 & \hspace*{-0.5cm} \textstyle \cdot (\hat b-2\hat A\sum_{j=1}^N [T^\nu]_{ij}x^j_{(k+1)} +\hat A\sum_{j=1}^N [T^\nu]_{ij}x^j_{(k)}  )] \label{eq:apa_inner_c},
\end{align*}
for all $i\in\Z[1,N]$ or, in compact form,
\begin{equation}
\begin{aligned}
& x_{(k+1)}\leftarrow \! \Pi_{\mathcal{X}}[x_{(k)} \! - \! \tau \left( F_\nu(x_{(k)}) + \! { A_\nu}^\top \lambda_{(k)} \right)],\\
& \lambda_{(k+1)}  \leftarrow \! \Pi_{\mathbb{R}^{Nm}_{\ge0}}[\lambda_{(k)}-\tau (b-2 A_\nu x_{(k+1)} + A_\nu x_{(k)} )]. 
\end{aligned}
\label{eq:alg_compact}
\end{equation}
As~\eqref{eq:alg_compact} coincides with one iteration of the asymmetric projection algorithm given in~\cite[Algorithm 2]{gentile2017nash} applied to VI$(\mathcal{Q}_\nu,F_\nu)$. 
Then~\cite[Theorem 3]{gentile2017nash} shows that, by choosing $\tau$ as in~\eqref{eq:condition_tau_APA}, which also implies $\tau<1/\| A_\nu\|$,
Algorithm~\ref{alg:NAG_nu} is guaranteed to converge to the unique solution of  VI$(\mathcal{Q}_\nu,F_\nu)$.
\end{proof}
\begin{remark}
Note that $A_\nu=T^\nu \otimes \hat A$, hence $\|A_\nu\|=\|T^\nu\|\|\hat A\|$.
Under Assumption~\ref{ass:primitive_doubly},
$\|T^\nu\|=1$ since $T$ is doubly stochastic. Hence in this case, one can 
use $\|\hat A\|$  instead of  $\|A_\nu\|$ in \eqref{eq:condition_tau_APA}.
Moreover, Lemma \ref{lemma:c3} guarantees that strong monotonicity of $F_\nu$ assumed in Theorem \ref{thm:convergence_alg} is met for $\nu>\nu_{\textup{SMON}}.$
\hfill $\square$
\end{remark}

\subsection{Relation of Theorems \ref{thm:convergence} and \ref{thm:convergence_alg}  with the  literature on distributed convergence in AAG  without coupling constraints}
\label{sec:literature_relation}
As anticipated in the introduction, distributed algorithms for AAGs  without coupling constraints have already been derived in the  literature, for example in~\cite{koshal2012gossip,chen2014autonomous,parise2015networkA,koshal2016distributed}.
We highlight here how the proofs of Theorems~\ref{thm:convergence} and~\ref{thm:convergence_alg} and the steps of Algorithm~\ref{alg:NAG_nu}
greatly simplify in the absence of coupling constraints (i.e. when $\mathcal{C}_\infty=\R^{Nn}$).

Regarding the first statement of Theorem \ref{thm:convergence},
in the absence of coupling constraints the VIs of $\mathcal{G}_\nu$ and  $\mathcal{G}_\infty$ feature the same set $\mc{X}$,
which is not affected by the parameter $\nu$.
Convergence of $\bar x_\nu$ to  $\bar x_\infty$ can thus be proven by using standard sensitivity analysis results for VI,
as explained in details at the beginning of Appendix B.
In this case one can even prove Lipschitz continuity of the solution,
so it is possible to derive bounds on the minimum number of communications $\nu$ needed to achieve any desired precision in~\eqref{eq:xnu_tends_xinf}.

Regarding the second statement of Theorem~\ref{thm:convergence}, in the absence of coupling constraints  the fact that $\bar x_\nu$ is an  $\varepsilon$-Nash equilibrium of $\mathcal{G}_\infty$ is a trivial consequence of \eqref{eq:xnu_tends_xinf} and of the fact that the cost functions are Lipschitz.
The difficulty when introducing the coupling constraints are that
i) the feasibility of $\bar x_\nu$ in $\mathcal{G}_\nu$ does not imply automatically feasibility of $\bar x_\nu$ in $\mathcal{G}_\infty$ and
ii) in the definition of Nash equilibrium, the set of feasible deviations $\mathcal{Q}^i_\nu(\bar x_\nu^{-i})$  in $\mathcal{G}_\nu$ is different from the set of feasible deviations $\mathcal{Q}^i_\infty(\bar x_\nu^{-i})$ in $\mathcal{G}_\infty$ (without coupling constraints both these sets would instead be simply $\mathcal{X}^i$).
This is why to prove the second statement of Theorem~\ref{thm:convergence} one needs to show Hausdorff convergence of $\mathcal{Q}^i_\nu(\bar x_\nu^{-i})$ to $\mathcal{Q}^i_\infty(\bar x_\nu^{-i})$ as $\nu \to \infty$, as done in Lemma~\ref{lemma:hausdorff} in the Appendix.

Regarding Algorithm~\ref{alg:NAG_nu} and Theorem \ref{thm:convergence_alg},
in the absence of coupling constraints one needs to solve the VI$(\mathcal{X},F_\nu)$ in a distributed fashion.
Since the constraint set $\mathcal{X}$ can be decoupled among the agents,
the standard projection algorithm \cite[Algorithm 12.1.1]{facchinei2007finite} is distributed and it is guaranteed to converge,
because $F_\nu$ is strongly monotone.
In other words, one can run Algorithm~\ref{alg:NAG_nu} performing only the primal steps, with simplified strategy update
$x^{i}_{(k+1)}\leftarrow \! \Pi_{\mathcal{X}^i}[x^{i }_{(k)} \! - \! \tau  F^i_{\nu,(k)} ]$.

\section{Application: Cournot game with transportation costs}
\label{sec:application} 

\subsection{Game definition}
\label{sec:appl}

Consider a single-commodity\footnote{The analysis applies also to a
multi-commodity game as in~\cite[Section 7.1]{yi2017distributed},
but we consider a single-commodity game for ease of exposition and we rather focus on the transportation costs.}
Cournot game with $N$ firms and $V$ markets,
which correspond to $V$ physical locations.
Firm $i\in\Z[1,N]$ chooses to sell $y_v^i \in \R_{\ge 0}$ amount of commodity at each market $v \in \Z[1,V]$.
Each firm $i$ produces its commodity at a given location $\ell_i\in\Z[1,V]$ and
then ships its commodity to the different markets over a transportation network,
where the $V$ nodes represent market locations and a
directed edge connecting two nodes represents a road connecting two markets.
We characterize the network by its incidence matrix $B\in\{0,1,-1\}^{V\times E}$,
where $E$ is the number of edges and $B_{v,e}=-1$ if edge $e$ leaves node $v$,
$B_{v,e}=1$ if edge $e$ enters node $v$ and $B_{v,e}=0$ otherwise.
Denote by $r^i\in\R_{\ge0}$ the total amount of commodity produced and sold by firm $i$ (i.e., $r\i=\sum_{v\in V} y_v^i$)
and by $t^i_e\in\R_{\ge0}$ the amount of commodity transported by firm $i$ over edge $e$, with $t^i=[t^i_e]_{e=1}^E$.
Define the strategy vector of firm $i$ as $x^i \defeq [t^i;r^i] \in \R_{\ge 0}^{E+1}$,
which uniquely determines $y\i \defeq [y^i_v]_{v=1}^V$, due to the balance equation
\begin{equation*}
y^i = Bt^i + e_{\ell_i} r\i = H\i x^i,
\label{eq:conservation_commodity}
\end{equation*}
with $H\i\defeq[B, e_{\ell_i}]\in\R^{V \! \times \! (E+1)}$ and $e_j$ the $j^\text{th}$ canonical vector.
\subsubsection{Cost function}
We assume that at each market the commodity is sold at a price that depends on the total commodity sold by the $N$ firms.
We allow for inter-market effects and use the inverse demand function\footnote{The inverse demand function determines the price for which demand equals supply at market $v$. This is why
 we can assume that all the supply is sold.} $p: \R^V_{\ge0}\rightarrow \R^V_{\ge0}$ that maps 
the normalized vector $\sinf(x) = \frac1N \sum_{j=1}^N y^j = \frac1N \sum_{j=1}^N H\j x^j$
to the vector of prices of each market $p(\sinf(x)) \defeq [p_v(\sinf(x))]_{v=1}^V$.
Then the revenue of firm $i$ is $p(\sinf(x))^\top y^i$.
Moreover, for firm $i$ transporting $t^i_e$ commodity over an edge comes with a cost equal to
%
\begin{equation}
c_e\i(t^i_e) \defeq \beta_{e}\i t^i_e- \gamma_e\i(t^i_e),
\label{eq:cost_transport}
\end{equation}
where, for all $i$, $\gamma_e\i$ is a strongly concave, increasing function  with maximum derivative smaller than $\beta_{e}\i$.
The transportation cost in~\eqref{eq:cost_transport} can be thought of as the sum of two terms:
the first is a cost proportional to the amount shipped,
the second term is a discount that increases as the amount of shipped commodity increases.

The production cost function of firm $i$ has a similar form
\begin{equation}
a\i(r^i) \defeq \beta_{a}\i r^i - \gamma_a\i(r^i),
\label{eq:cost_production}
\end{equation}
where $\gamma_a\i$ is a strongly concave, increasing function with maximum derivative smaller than $\beta_{a}\i$.
Note that the functions~\eqref{eq:cost_transport} and~\eqref{eq:cost_production} are strongly convex,
as in~\cite[Section 1.4.3]{facchinei2007finite}.

To sum up, the cost function of firm $i$ is
\begin{equation*}
J^i(x^i,\sinf(x)) \defeq \!\!\!
\underbrace{a\i(r^i)}_{\textup{production cost}} \! + \underbrace{\textstyle \sum_{e=1}^E c_e\i(t^i_e)}_{\textup{transportation cost}}\; - \; \underbrace{p(\sinf(x))^\top y^i}_{\textup{revenue}} \! .
\end{equation*}

\subsubsection{Constraints}
The strategy of firm $i$ must satisfy the individual constraints
\begin{equation}
\mathcal{X}^i \defeq \{x^i\in\R^{E+1}_{\ge 0} \vert x\i \le \bar{r}^i \cdot \ones[E+1], y\i = H\i x\i \ge 0 \},
\label{eq:constraints_application}
\end{equation}
where $\bar{r}^i$ is the production capacity of firm $i$.
Note that ~\eqref{eq:constraints_application} implies $t\i_e \le \bar r\i$ for each $e$,
which is needed to guarantee boundedness of $\X\i$ and can be imposed without loss of generality,
as the transportation costs $c\i_e$ are increasing.

Moreover, we assume that each market $v$ is composed by retailers whose storage capacity imposes an upper bound $K_v$ on the total commodity that can be sold at market $v$, thus giving rise to the coupling constraints $\sigma_\infty(x) \le K\defeq[K_v]_{v=1}^V$.


\subsubsection{Communication network}
We assume that the firms can communicate with each other according to a sparse communication network, described by the adjacency matrix $T$,
which we assume satisfies Assumption \ref{lemma:lim_A_nu_F_nu}.
This network can model spatial proximity of firms,
or the fact that they may want to share their strategies only with firms they trust.

\subsection{Theoretical guarantees}

The cost, constraints and communication network introduced above give rise to a game as in~\eqref{eq:GNEP}, with the only difference that the aggregate $\sigma_\infty(x)$ depends on $y^i = H\i x^i$ instead of $x^i$ directly. We next show that our theory can be easily extended to cover such case.\footnote{Defining a new game with strategies
$\tilde x\i = H\i x\i$, so that the aggregate depends only on  the $\tilde x\i$, is not possible in general as the cost $J\i(x\i,\sigma_\infty(x))$ cannot  be expressed as a function of the variables $\tilde x\i$s
unless $H\i$ is full column rank for all $i$. This is not the case in the Cournot game under consideration.}

\subsubsection{Extension}
\label{sec:ext}
Set $H_\text{blkd} \defeq\mbox{blkdiag}(H^1,\ldots,H^N) \in \R^{NV \times N(E+1)}$ .
The quantities in~\eqref{eq:game_aag} relative to $\Ginf$ are
%
%
\begin{subequations}\label{eq:game_aag_dim}
\begin{align}
\Finf(x) &\defeq [\nabla_{x^i} J^i(x^i, \sinf(x) )]_{i=1}^N, \label{eq:game_aag_dim:a}\\
&= \textstyle [\nabla_{z_1} J^i(x^i,\sigma_\aag(x))+\frac1N H_i^\top \nabla_{z_2} J^i(x^i,\sigma_\aag(x))]_{i=1}^N, \notag \\
\mathcal{Q}_\aag&\defeq\{x\in \X^1 \times \dots \times \X^N \vert A_\aag x \le b\}, \label{eq:game_aag_dim:b} \\
A_\aag&\defeq \biggl(\frac{1}{N} \ones[N]\ones[N]^\top \otimes \hat A \biggr) H_\text{blkd}, \label{eq:game_aag_dim:c}\\
b&\defeq\ones[N] \otimes \hat b.
\label{eq:game_aag_dim:d} 
\end{align}
\end{subequations}
The quantities in~\eqref{eq:game_nu_nag} relative to $\Gnu$ are
\begin{subequations}\label{eq:game_nu_nag_dim}
\begin{align}
F_\nu(x) &\defeq [\nabla_{x^i} J^i(x^i, \snu^i(x) )]_{i=1}^N \label{eq:game_nu_nag_dim:a},\\
&= \textstyle [\nabla_{z_1} J^i(x^i,\sigma^i_\nu(x))+ [T^\nu]_{ii}  H_i^\top \nabla_{z_2} J^i(x^i,\sigma^i_\nu(x))]_{i=1}^N,\notag \\
\mathcal{Q}_\nu &\defeq\{x\in \X^1 \times \dots \times \X^N \vert A_\nu x \le b \}, \label{eq:game_nu_nag_dim:b} \\
A_\nu &\defeq(T^\nu \otimes \hat A) H_\text{blkd}. \label{eq:game_nu_nag_dim:c}
\end{align}
\end{subequations}
%

%
The only assumption of Sections~\ref{sec:agg_games}-\ref{sec:auxiliary_results} that needs to be modified is Assumption~\ref{ass:R2},
which, due to the presence of $H_\text{blkd}$ in~\eqref{eq:game_aag_dim:c}, cannot be expressed on the sole $\hat A$ and $\hat b$ and is hence replaced by the following one.

\textbf{Assumption 2\textup{'}} (Coupling constraints).
The matrix $\Ainf$ in~\eqref{eq:game_aag_dim:c} and the vector $b$ in~\eqref{eq:game_aag_dim:d}
are such that the following implication holds.
\begin{align*}
\{\Ainf^\top  s= 0 ,\quad  b^\top s \le 0, \quad   s\ge 0  \} \quad \Rightarrow \quad s=0.
\tag*{$\square$}
\end{align*}
The proofs of Theorems~\ref{thm:main} and~\ref{thm:convergence}, which made use of Assumption~\ref{ass:R2},
can be conducted in the same way using Assumption~2'.

\subsubsection{Verify the assumptions}

To use Theorem \ref{thm:main}  in the numerical analysis of the next subsection,
we need to verify its assumptions. 
Assumption~\ref{lemma:lim_A_nu_F_nu} holds by problem statement.
Assumption~2' is satisfied, because the implication $\{(\ones[N] \otimes K)^\top s \le 0, s \ge 0\} \Rightarrow s = 0$
is satisfied, given that $K > 0$. 
To guarantee that Assumption~\ref{ass:F_Fnu_SMON} holds we make the following  assumption, whose sufficiency is  proven in Lemma \ref{lemma:verifyAss3}.

\begin{assumption}[Cournot-game regularity conditions]\label{ass:appl}
 The cost $J\i(z_1,z_2)$ is twice continuously differentiable for all $i$,
and  the inverse demand function $p$ satisfies one of the following conditions.
\\
1)  $p$ is affine, i.e., $p(\sigma_\infty(x))= -D \sigma_\infty(x) +d$, for some $D\in\R^{V\times V}$, $d\in\R^V$  and $D\succeq 0$.
\\
2)  $p_v$  depends only on the commodity sold at $v$, i.e.,
$p(\sigma_\infty(x))\eqdef [p_v([\sigma_\infty(x)]_v)]_{v=1}^V$.
For each $v$,  $p_v$  is twice continuously differentiable, strictly decreasing and satisfies 
\begin{equation}
\label{eq:price_condition}
\min_{\substack{v \in \{1,\dots,V\} \\ z\in [0, \tilde r]}} \left( -p'_v(z)+\frac{\tilde r}{8}p''_v(z) \right)>0,\quad \tilde r \defeq \!\! \maxx{i \in \Z[1,N] } \bar r^i.
\end{equation}
\hfill $\square$
\end{assumption}

\begin{lemma}
If Assumption \ref{ass:appl} holds then the Cournot game in Section \ref{sec:appl} satisfies Assumption 3. \hfill $\square$
\label{lemma:verifyAss3}
\end{lemma}
\noindent The proof is reported in Appendix C.

\begin{remark}
If the function $p$ is as in Assumption~\ref{ass:appl}.1) and $D=D^\top$,
then $\Ginf$ is a potential game~\cite{shaply1994potentialgames}.
In other words, there exists a function $f: \mc{Q} \to \R$ such that $\nabla_x f(x) = F_\infty(x)$ and
VI($\mc Q_\infty$,$F_\infty$) is equivalent to $\argmin{x \in \mc{Q_\infty}}{f(x)}$,
as described in~\cite[Section 1.3.1]{facchinei2007finite}.
Then a Nash equilibrium can be found by solving the optimization program $\argmin{x \in \mc{Q_\infty}}{f(x)}$.
We also note that condition~\eqref{eq:price_condition} is satisfied if for each market $v$ the function $p_v$ is convex and strictly decreasing.
\hfill $\square$
\end{remark}
Regarding the Standing Assumption, Assumption~\ref{ass:F_Fnu_SMON} implies that $J\i$ is  continuously differentiable in its arguments
and that $\nabla_x[\nabla_{x\i} J\i(x\i,\sinf(x))]_{i=1}^N \succ \alpha I_{N(E+1)}$ by Proposition~\ref{prop:pd},
which in turn implies
$\nabla_{x\i} \nabla_{x\i} J\i(x\i,\sinf(x)) \succ \alpha I_{E+1}$, which implies convexity of $J\i$ in $x\i$ for all fixed $x^{-i}$. The sets $\mathcal{X}^i$ are trivially convex, compact and non-empty.

Finally, the next lemma shows that for the network, number of communications $\nu$ and price functions used in the numerical analysis of the next subsection, $F_\nu$ is strongly monotone, as required by Theorem~\ref{thm:convergence_alg}.
\begin{lemma}\label{lem:4}
Under Assumption \ref{ass:appl}.1), if $T$ is the adjacency matrix of an undirected network, so that $T=T^\top$,
then the operator $F_\nu$ is strongly monotone for any $\nu$ even.
\hfill $\square$
\end{lemma}
\noindent The proof is  reported in Appendix C.

\subsection{Numerical analysis}
We consider two simulation setups.
The first is a small example to develop intuition about the problem,
the second is used to illustrate the applicability of our method for a more realistic scenario.

\subsubsection{Small network}
We consider a simple chain transportation network with $V=5$ markets, $E=4$ roads and $N=3$ firms.
As illustrated in Figure \ref{fig:net} we assume that the firms $\{1,2,3\}$ are located at markets $\{1,3,5\}$, respectively, and are otherwise identical, with $\bar r\i = 5$ for all $i$.
\begin{figure}[H]
\begin{center}
\includegraphics[width=0.45\textwidth]{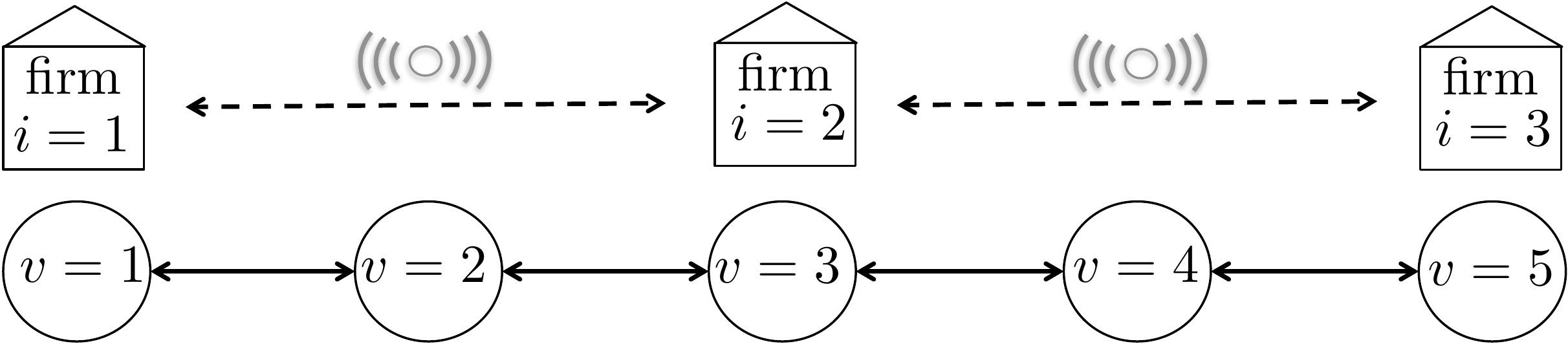}
\end{center}
\caption{Small  network with the $3$ firms located in markets $1,3,5$, respectively. The  transportation network is represented with a solid line, the communication network with a dashed line.}
\label{fig:net}
\end{figure}
\noindent
Regarding the cost functions, for each firm $i$, we set

\begin{subequations}
\begin{align}
c\i_e(t\i_e) &= c_e(t\i_e)=t\i_e-\left(1-\frac{1}{1+t\i_e}\right) \forall e, \label{eq:small_network_c_e} \\
a\i(r\i) &= a(r\i)=2\left[r\i-\left(1-\frac{1}{1+r\i}\right)\right] \! . \label{eq:small_network_a_p}
\end{align}
\end{subequations}

For each market $v$,
we consider the inverse demand function $p_v$ to be affine and independent from the commodity sold at other markets.
Specifically, for all $v$, $p_v(\sigma) = 10-\sigma_v$.
We assume that firm $2$ bidirectionally communicates with firms $1$ and $3$, while $1$ and $3$ do not communicate, according to the communication matrix
$$T=\left[\begin{array}{ccc}2/3 & 1/3 & 0 \\1/3 & 1/3 & 1/3 \\0 & 1/3 & 2/3\end{array}\right],$$
which is primitive and doubly stochastic, hence satisfies Assumption \ref{ass:primitive_doubly}.

We run Algorithm~\ref{alg:NAG_nu} with $\tau=0.005$, $\nu=10$ and initial conditions all equal to zero\footnote{The values in~\eqref{eq:condition_tau_APA} can be shown to be $\alpha_\nu = 4/(1+\tilde r)^3 = 0.0185$,
$L_\nu = \lambda_\text{max}(H_\text{blkd}^\top [(I_N \otimes D)(T^\nu \otimes I_E) + \text{diag}(T^\nu)\otimes D^\top] H_\text{blkd}) = 9.9124$
and $\| A_\nu \| = 1$;
then~\eqref{eq:condition_tau_APA} reads $\tau < 1.8 \cdot 10^{-4}$.
This is a conservative bound, we verified by simulations that  the algorithm converges also for $\tau = 0.005$.}.
We use $\max\{\|x_{(k)}-x_{(k-1)}\|_\infty,\|\lambda_{(k)}-\lambda_{(k-1)}\|_\infty\}< 10^{-4}$ as stopping criterion.
Figure~\ref{fig:y} (top) reports
the sales $y^i$ for each firm in the $5$ markets at the variational Nash equilibrium of $\Gnu$ (with $\nu=10$),
for the case when there are no coupling constraints (i.e., $K$ is chosen so large that it has no effect).
Figure~\ref{fig:y} (bottom) reports how the equilibrium changes if we introduce the coupling constraint $[\sigma_\infty(x)]_3 \le 1/3$, so that the total capacity of market $3$ is $1$.
In both cases the $\G_\nu$ (with $\nu=10$) variational equilibrium is an $\varepsilon_\nu$-Nash equilibrium for $\Ginf$,
as by the second statement of Theorem~\ref{thm:convergence}.
The value of $\varepsilon_\nu$ can be computed after convergence according to Definition~\ref{def:NE}.
A more descriptive quantity is the relative maximum improvement $\hat \varepsilon_\nu$, defined as\footnote{Note that for any fixed $\bar x_\nu$, $\hat \varepsilon_\nu$ in \eqref{eq:relative_eps} can be computed by solving the $N$ optimization problems $\{ \min_{ x\i \in \mc{Q}\i_\infty(\bar x_\nu^{-i})} J\i(x\i,\frac{1}{N}x\i \! + \! \sum_{j\neq i}\frac1N \bar x_\nu\j)\}_{i=1}^N$. }
\begin{equation}
\hat \varepsilon_\nu \defeq \!\!\!\!\!\! \maxx{i, x\i \in \mc{Q}_{\infty}\i(\bar x_\nu^{-i})}{\!\!\!\!\!\!\!\! \frac{\! J\i(\bar x\i_\nu,\sigma_\infty(\bar x_\nu))-J\i(x\i,\frac{1}{N}x\i \! + \! \sum_{j\neq i}\frac{1}{N} \bar x_\nu\j) \! }{J\i(\bar x\i_\nu,\sigma_{\infty}(\bar x_\nu))}},
\label{eq:relative_eps}
\end{equation}
which equals $0.0014$ (for the game without coupling constraint) and $0.0035$ (for the game  with coupling constraint).
%
\begin{figure} [h!]
\begin{center}
\newlength\figureheight 
\newlength\figurewidth 
\setlength\figureheight{1.5cm} 
\setlength\figurewidth{8cm} 
%
%
\definecolor{mycolor1}{rgb}{0.13725,0.36078,0.43922}%
\definecolor{mycolor2}{rgb}{0.09020,0.58039,0.69020}%
\definecolor{mycolor3}{rgb}{0.56078,0.72941,0.00000}%
\begin{tikzpicture}

\begin{axis}[%
width=\figurewidth,
height=\figureheight,
at={(1.011in,0.295in)},
scale only axis,
bar shift auto,
xmin=0.5,
xmax=5.5,
xtick={1,2,3,4,5},
xticklabels={,,,,},
ymin=0,
ymax=3.5,
axis background/.style={fill=white},
xmajorgrids,
font=\footnotesize,
ymajorgrids,
]
\addplot[ybar, bar width=0.178, fill=mycolor1, draw=black, area legend] table[row sep=crcr] {%
1	3.29425274158038\\
2	1.49205180615129\\
3	0.114078278385911\\
4	0.0995819408862545\\
5	3.52156456322455e-05\\
};
\addplot[forget plot, color=white!15!black] table[row sep=crcr] {%
0.5	0\\
5.5	0\\
};

\addplot[ybar, bar width=0.178, fill=mycolor2, draw=black, area legend] table[row sep=crcr] {%
1	0.0442578236052739\\
2	1.02552970067784\\
3	2.8604249348974\\
4	1.02552970067784\\
5	0.044257823605274\\
};
\addplot[forget plot, color=white!15!black] table[row sep=crcr] {%
0.5	0\\
5.5	0\\
};

\addplot[ybar, bar width=0.178, fill=mycolor3, draw=black, area legend] table[row sep=crcr] {%
1	3.52156456322455e-05\\
2	0.0995819408862544\\
3	0.114078278385911\\
4	1.49205180615129\\
5	3.29425274158038\\
};
\addplot[forget plot, color=white!15!black] table[row sep=crcr] {%
0.5	0\\
5.5	0\\
};

\end{axis}
\end{tikzpicture}%
%
%
\definecolor{mycolor1}{rgb}{0.13725,0.36078,0.43922}%
\definecolor{mycolor2}{rgb}{0.09020,0.58039,0.69020}%
\definecolor{mycolor3}{rgb}{0.56078,0.72941,0.00000}%
\begin{tikzpicture}

\begin{axis}[%
width=\figurewidth,
height=\figureheight,
at={(1.011in,0.295in)},
scale only axis,
bar shift auto,
xmin=0.5,
xmax=5.5,
xtick={1,2,3,4,5},
xticklabels={{market 1},{market 2},{market 3},{market 4},{market 5}},
ymin=0,
font=\footnotesize,
ymax=3.5,
axis background/.style={fill=white},
xmajorgrids,
ymajorgrids,
legend style={at={(0.335,0.43)}, anchor=south west, legend cell align=left, align=left, draw=white!15!black, font=\footnotesize},
]
\addplot[ybar, bar width=0.178, fill=mycolor1, draw=black, area legend] table[row sep=crcr] {%
1	3.44356081028404\\
2	1.41722429616991\\
3	3.61474362928216e-06\\
4	0.137946778993244\\
5	0.00126433281767803\\
};
\addplot[forget plot, color=white!15!black] table[row sep=crcr] {%
0.5	0\\
5.5	0\\
};

\addplot[ybar, bar width=0.178, fill=mycolor2, draw=black, area legend] table[row sep=crcr] {%
1	0.248983568785253\\
2	1.73600789535927\\
3	1.03001689461111\\
4	1.73600789535927\\
5	0.248983568785253\\
};
\addplot[forget plot, color=white!15!black] table[row sep=crcr] {%
0.5	0\\
5.5	0\\
};

\addplot[ybar, bar width=0.178, fill=mycolor3, draw=black, area legend] table[row sep=crcr] {%
1	0.00126433281767803\\
2	0.137946778993244\\
3	3.61474362928151e-06\\
4	1.41722429616991\\
5	3.44356081028404\\
};
\addplot[forget plot, color=white!15!black] table[row sep=crcr] {%
0.5	0\\
5.5	0\\
};

\end{axis}
\end{tikzpicture}%
\caption{Production per firm and market without coupling constraints (top)
and with the coupling constraint $[\sigma\i(x)]_3 \le 1/3$ (bottom).
In both cases the total production at the equilibrium is $r\i=\bar r^i= 5$ for all $i$.
Both simulations are obtained with $\nu=10$ communications.
}
\label{fig:y}
\end{center}
\end{figure}
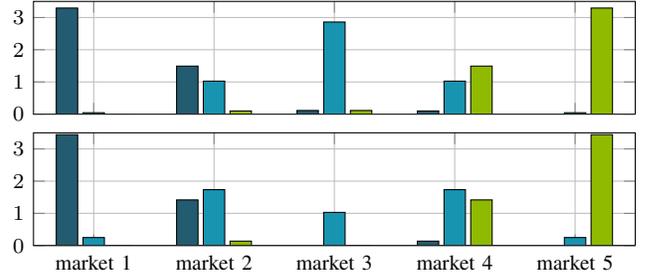

\subsubsection{Large network} \label{sec:large}
As a more realistic example we consider the transportation network illustrated in Figure \ref{fig:net_large} which consists of
$V=43$ possible markets and $E=51$ (bidirectional) edges connecting them.
The network is taken from the data set~\cite{OldenburgDataset}, which provides also the Cartesian coordinates of the vertexes.
We consider $5$ firms that differ only for their locations $\ell\i$,
which are $\ell_1 = 37$, $\ell_2 = 20$, $\ell_3 = 11$, $\ell_4 = 6$, $\ell_5 = 35$
as indicated in Figure~\ref{fig:net_large}.
Each firm has a production capacity of $\bar r\i = 10$, while we consider a  capacity of $1.5$ for each market (i.e. $K=1.5/5$).
The production cost is as in~\eqref{eq:small_network_a_p}
while the transportation cost for edge $e$ is the same for each firm $i$ and is 
\begin{equation*}
c\i_e(t\i_e) = c_e(t\i_e)= \rho_e \left(t\i_e-\left(1-\frac{1}{1+t\i_e}\right)\right),
\end{equation*}
where $\rho_e\in]0,1]$ is the normalized\footnote{The normalized length of a road is defined as the absolute length divided by the maximum length road in the network.} length of road $e$.
The inverse demand function $p$ is affine, i.e. $p(\sigma) = 10 \cdot \ones[45] - D \sigma$ and it encodes intra-market competition via the matrix $D$ whose component in position $(h,k)$ is $[D]_{h,k}=1$ if $h=k$,
$[D]_{h,k}=0.3(1-\rho_e)$, if there is a road $e=(h,k)$ between markets $h$ and $k$, while $[D]_{h,k}=0$ otherwise.
In words, the price $p_v$ at market $v$ not only decreases when more commodity is sold at $v$,
but also when more commodity is sold at the neighboring markets, with physically close markets being more influential.
We verified numerically that $D \succeq 0$.
We use the communication matrix $T$ that corresponds to a symmetric ring, i.e.,
\begin{equation*}
T = \begin{bmatrix}
0 \!\! & \!\! 0.5 \!\! & \!\! 0 \!\! & \!\! 0 \!\! & \!\! 0.5 \\[-0.05cm]
0.5 \!\! & \!\! 0 \!\! & \!\! 0.5 \!\! & \!\! 0 \!\! & \!\! 0 \\[-0.05cm]
0 \!\! & \!\! 0.5 \!\! & \!\! 0 \!\! & \!\! 0.5 \!\! & \!\! 0 \\[-0.05cm]
0 \!\! & \!\! 0 \!\! & \!\! 0.5 \!\! & \!\! 0 \!\! & \!\! 0.5 \\[-0.05cm]
0.5 \!\! & \!\! 0 \!\! & \!\! 0 \!\! & \!\! 0.5 \!\! & \!\! 0
\end{bmatrix},
\end{equation*}
which satisfies Assumption~\ref{ass:primitive_doubly}.
We run Algorithm~\ref{alg:NAG_nu} with $\tau=0.05$, initial conditions all equal to zero and different values of $\nu$\footnote{The values in~\eqref{eq:condition_tau_APA} can be shown to be
$\alpha_\nu = 4/(1+\tilde r)^3 = 0.003$,
$L_\nu = \lambda_\text{max}(H_\text{blkd}^\top [(I_N \otimes D)(T^\nu \otimes I_E) + \text{diag}(T^\nu)\otimes D^\top] H_\text{blkd}) = 12.89$
and $\| A_\nu \| = 1$;
then~\eqref{eq:condition_tau_APA} reads $\tau < 1.8 \cdot 10^{-5}$. This is a conservative bound, we verified by simulations that  the algorithm converges also for $\tau = 0.05$.}.
As for the small network, we use $\max\{\|x_{(k)}-x_{(k-1)}\|_\infty,\|\lambda_{(k)}-\lambda_{(k-1)}\|_\infty\}< 10^{-4}$ as stopping criterion.
We consider even values of $\nu$ between $2$ and 20.
For each $\nu$ we run Algorithm~\ref{alg:NAG_nu} and find the variational Nash equilibrium of $\Gnu$,
which is an $\varepsilon_\nu$-Nash equilibrium for $\Ginf$,
as by the second statement of Theorem~\ref{thm:convergence}.
After convergence $\varepsilon_\nu$ can be computed according to Definition~\ref{def:NE} and $\hat \varepsilon_\nu$ according to~\eqref{eq:relative_eps}.
Figure~\ref{fig:nu} (top) reports the value of $\hat \varepsilon_\nu$ as function of $\nu$,
thus numerically verifying the second statement of Theorem~\ref{thm:convergence}.
Figure~\ref{fig:nu} (bottom) reports the value of $\| \bar x_\nu - \bar x_\infty \|_2$ as function of $\nu$,
thus numerically verifying the first statement~\eqref{eq:xnu_tends_xinf} of Theorem~\ref{thm:convergence}.
In Figure \ref{fig:net_large} we illustrate the variational Nash equilibrium of $\Ginf$
obtained by setting $\nu=1$ and $T=\frac{1}{N}\ones[N]\ones[N]^\top$.
We note that each firm is the only seller at the location where it produces,
and more in general firms tend to sell close to their production location, as expected.

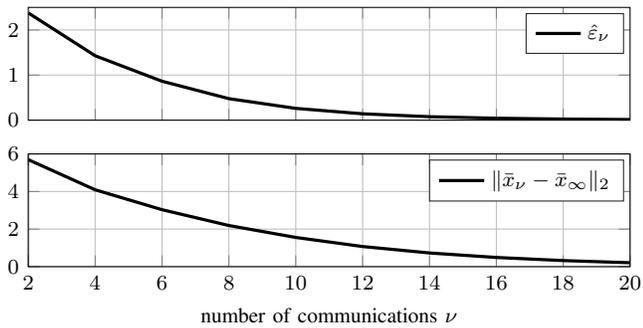
\begin{figure}[h]
\begingroup
\centering
\setlength\figureheight{1.5cm} 
\setlength\figurewidth{8cm} 
%
%
\definecolor{mycolor1}{rgb}{0.00000,0.44700,0.74100}%
\begin{tikzpicture}

\begin{axis}[%
width=\figurewidth,
height=\figureheight,
at={(1.011in,0.642in)},
scale only axis,
xmin=2,
xmax=20,
xlabel style={font=\color{white!15!black}},
xtick={4,6,8,10,12,14,16,18},
xticklabels={},
ymin=0,
ymax=2.5,
ylabel style={font=\color{white!15!black}},
axis background/.style={fill=white},
xmajorgrids,
ymajorgrids,
font=\footnotesize,
legend style={at={(0.99,0.95)}, legend cell align=left, align=left, draw=white!15!black}
]
\addplot [color=black, line width=1.2pt]
  table[row sep=crcr]{%
2	2.37825466840761\\
4	1.42839666384774\\
6	0.863306632010472\\
8	0.477179216315378\\
10	0.261496779023576\\
12	0.141008500899686\\
14	0.0768126086920488\\
16	0.042892832233248\\
18	0.0246255797134968\\
20	0.0140885604869569\\
};
\addlegendentry{$\hat \varepsilon_\nu$}

\end{axis}
\end{tikzpicture}
%
%
\begin{tikzpicture}

\begin{axis}[%
width=\figurewidth,
height=\figureheight,
at={(1.011in,0.642in)},
scale only axis,
xmin=2,
xmax=20,
xlabel style={font=\color{white!15!black},font=\footnotesize,},
xlabel={$\text{number of communications } \nu$},
xtick={2,4,6,8,10,12,14,16,18,20},
ymin=0,
ymax=6,
ylabel style={font=\color{white!15!black}},
axis background/.style={fill=white},
xmajorgrids,
ymajorgrids,
font=\footnotesize,
legend style={at={(0.99,0.95)}, legend cell align=left, align=left, draw=white!15!black}
]
\addplot [color=black, line width=1.2pt]
  table[row sep=crcr]{%
2	5.69373801994183\\
4	4.08972385359315\\
6	3.03158336306071\\
8	2.18972654373934\\
10	1.55507808149399\\
12	1.07202528420688\\
14	0.729371482391083\\
16	0.48955103094438\\
18	0.324540624122954\\
20	0.208989648690819\\
};
\addlegendentry{$\| \bar x_\nu - \bar x_\infty \|_2$}

\end{axis};
\end{tikzpicture}%
\endgroup
\caption{Relative maximum cost improvement $\hat \varepsilon_\nu$ given in~\eqref{eq:relative_eps} as a function of $\nu$ for the large network.
}
\label{fig:nu}
\end{figure}

\begin{figure}[h]
\begingroup
\centering
\setlength\figureheight{2cm} 
\setlength\figurewidth{8.3cm} 
\hspace{0.4cm}\includegraphics[width=0.45\textwidth]{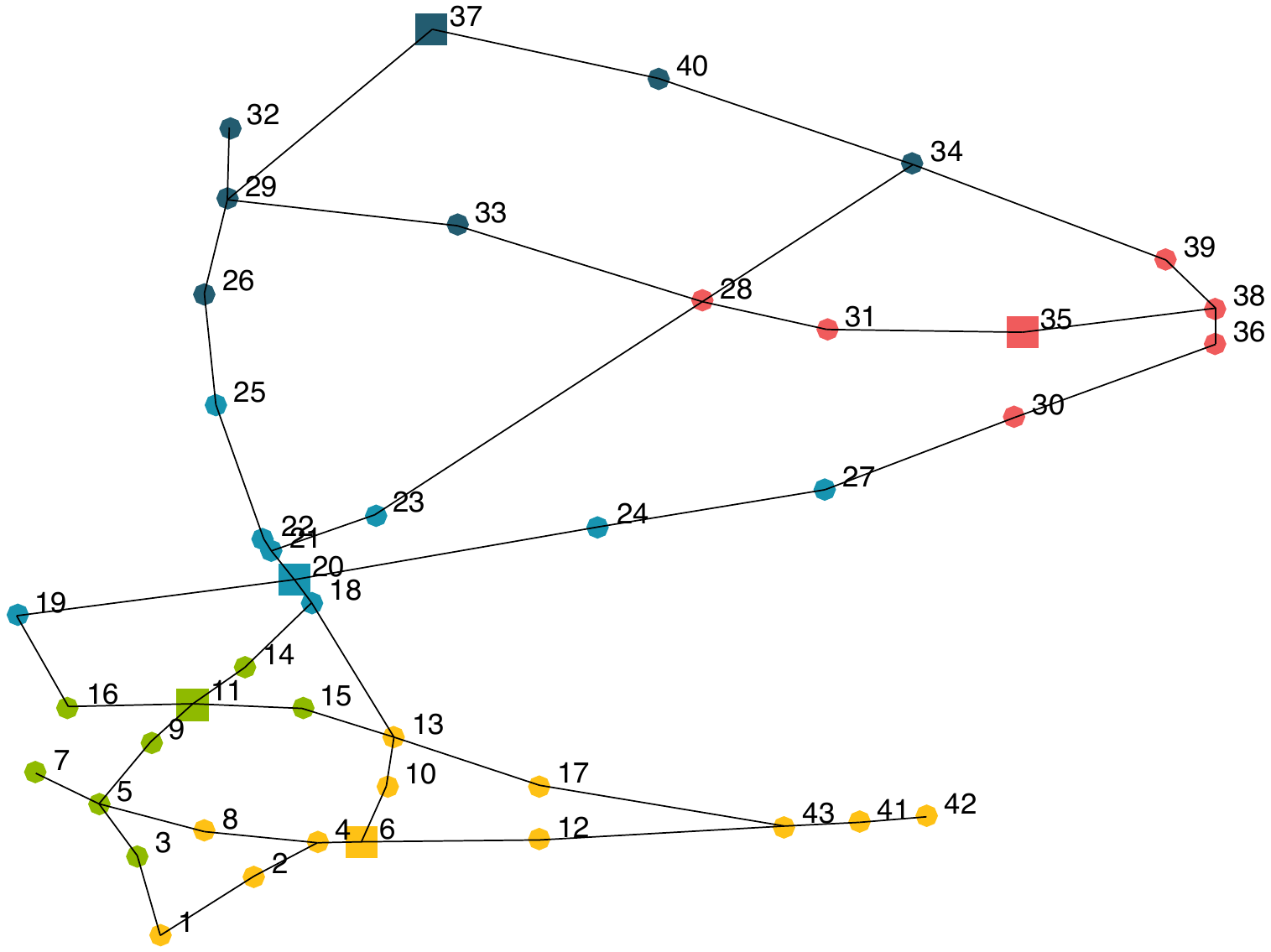}\\[0.2cm]
%
%
\definecolor{mycolor1}{rgb}{0.13725,0.36078,0.43922}%
\definecolor{mycolor2}{rgb}{0.09020,0.58039,0.69020}%
\definecolor{mycolor3}{rgb}{0.56078,0.72941,0.00000}%
\definecolor{mycolor4}{rgb}{0.99608,0.75686,0.08235}%
\definecolor{mycolor5}{rgb}{0.94118,0.35686,0.36078}%
\begin{tikzpicture}

\begin{axis}[%
width=\figurewidth,
height=\figureheight,
at={(1.011in,0.721in)},
scale only axis,
bar width=0.8,
xmin=0,
xmax=44,
ymin=0,
ymax=2.6,
ylabel style={font=\color{white!15!black}},
ylabel style={rotate=-90},
xlabel={$v$},
every axis x label/.style={
    at={(ticklabel* cs:1.01)},
    anchor=near ticklabel,
},
axis background/.style={fill=white},
font=\footnotesize,
xmajorgrids,
ymajorgrids,
legend style={legend columns=-1, at={(0.007,0.68)}, anchor=south west, legend cell align=left, align=left, draw=white!15!black, font=\footnotesize},
]

\addplot [color=black, line width=0.38pt, dashed]
  table[row sep=crcr]{%
1	1.5 \\
2	1.5 \\
3	1.5 \\
4	1.5 \\
5	1.5 \\
6	1.5 \\
7	1.5 \\
8	1.5 \\
9	1.5 \\
10	1.5 \\
11	1.5 \\
12	1.5 \\
13	1.5 \\
14	1.5 \\
15	1.5 \\
16	1.5 \\
17	1.5 \\
18	1.5 \\
19	1.5 \\
20	1.5 \\
21	1.5 \\
22	1.5 \\
23	1.5 \\
24	1.5 \\
25	1.5 \\
26	1.5 \\
27	1.5 \\
28	1.5 \\
29	1.5 \\
30	1.5 \\
31	1.5 \\
32	1.5 \\
33	1.5 \\
34	1.5 \\
35	1.5 \\
36	1.5 \\
37	1.5 \\
38	1.5 \\
39	1.5 \\
40	1.5 \\
41	1.5 \\
42	1.5 \\
43	1.5 \\
};
\addlegendentry{upper bound \,};

\addplot[ybar stacked, fill=mycolor1, draw=black, area legend] table[row sep=crcr] {%
1	0.00296838996673355\\
2	0.00272066052544265\\
3	0.00301753048547747\\
4	0.00252553149211568\\
5	0.00282506538329901\\
6	0.0025417301448517\\
7	0.00303573334327925\\
8	0.00275308438934816\\
9	0.00272051924579048\\
10	0.00287805161972066\\
11	0.00261916030848974\\
12	0.00332931421851783\\
13	0.00305999461177792\\
14	0.00282156694952471\\
15	0.00289128762281998\\
16	0.00322352041202714\\
17	0.00386881272794273\\
18	0.00331179336950194\\
19	0.00399990782113004\\
20	0.00400603988232238\\
21	0.00606322130260498\\
22	0.00679135250001976\\
23	0.0330615647753374\\
24	0.010003908866891\\
25	0.373222674196518\\
26	1.09803182183648\\
27	0.0752521231244626\\
28	0.175135217410455\\
29	1.42639903928355\\
30	0.0865941037985999\\
31	0.00798608272995371\\
32	1.3982359690942\\
33	0.900698275200258\\
34	1.14231208178992\\
35	0.00242470042314227\\
36	0.0116162903296386\\
37	1.49524501727814\\
38	0.0080584984981005\\
39	0.131104853449724\\
40	1.49061173983537\\
41	0.00797815902660749\\
42	0.0474946363043118\\
43	0.00451031927195992\\
};
\addlegendentry{$y^1_v \;\;\!$};
\addplot[forget plot, color=white!15!black] table[row sep=crcr] {%
0	0\\
45	0\\
};
\addplot[ybar stacked, fill=mycolor2, draw=black, area legend] table[row sep=crcr] {%
1	0.060726827074943\\
2	0.011879701665274\\
3	0.104648491198276\\
4	0.00768368053556144\\
5	0.0192289953225505\\
6	0.00699904678463413\\
7	0.127692664873805\\
8	0.0126364803304252\\
9	0.0965495754203081\\
10	0.0585376920969138\\
11	0.0210863325112013\\
12	0.0127812928999279\\
13	0.100463760836057\\
14	0.389352092270931\\
15	0.180449273801115\\
16	0.21253308340232\\
17	0.234387667048454\\
18	1.23640001335092\\
19	0.626735167371896\\
20	1.3646150013885\\
21	0.811079619503209\\
22	1.38938079180193\\
23	1.08844694887495\\
24	0.695634625009926\\
25	0.478819415581024\\
26	0.0681462871695456\\
27	0.216416795959673\\
28	0.0124594540156075\\
29	0.00649244044192526\\
30	0.0495868452977229\\
31	0.0046002978795681\\
32	0.00956305940388625\\
33	0.0418313774412897\\
34	0.00569018826742583\\
35	0.00202848104108673\\
36	0.00635181221394567\\
37	0.00119110299802443\\
38	0.00423696373020194\\
39	0.00475959499832154\\
40	0.00235013211147471\\
41	0.0759375318626747\\
42	0.0792721251927081\\
43	0.060288863252403\\
};
\addlegendentry{$y^2_v \;\;\!$};
\addplot[forget plot, color=white!15!black] table[row sep=crcr] {%
0	0\\
45	0\\
};
\addplot[ybar stacked, fill=mycolor3, draw=black, area legend] table[row sep=crcr] {%
1	0.171027202433406\\
2	0.0489688288465771\\
3	0.510071706798914\\
4	0.0120081880484357\\
5	0.277121484880068\\
6	0.0094983785031716\\
7	0.641353953982212\\
8	0.283666594808447\\
9	1.35376999451956\\
10	0.20118206712854\\
11	0.86837038404007\\
12	0.0281889324276827\\
13	0.275838537504501\\
14	1.0961964970205\\
15	1.13653884425222\\
16	1.15606516704631\\
17	0.343879698718563\\
18	0.151839336273415\\
19	0.486371182919066\\
20	0.0291992779066407\\
21	0.049513676992865\\
22	0.0882515941775989\\
23	0.148443020509218\\
24	0.149321378354831\\
25	0.087781467596395\\
26	0.00661417260026264\\
27	0.0774222790709668\\
28	0.00426616810828266\\
29	0.00347470042322792\\
30	0.00887530629908401\\
31	0.00295712896850128\\
32	0.00414549180148418\\
33	0.00452584129336795\\
34	0.00343888233892502\\
35	0.00169075777043915\\
36	0.0040360365790934\\
37	0.00112372070198537\\
38	0.00308648573374495\\
39	0.00329263769012039\\
40	0.00187576617260611\\
41	0.0913695946524459\\
42	0.0895786915073126\\
43	0.0837108909987597\\
};
\addlegendentry{$y^3_v \;\;\!$};
\addplot[forget plot, color=white!15!black] table[row sep=crcr] {%
0	0\\
45	0\\
};
\addplot[ybar stacked, fill=mycolor4, draw=black, area legend] table[row sep=crcr] {%
1	0.566899274781912\\
2	1.11970309626613\\
3	0.369992852414854\\
4	0.989766412948579\\
5	0.103125747887453\\
6	1.36179595340488\\
7	0.392648042401582\\
8	0.992790289105328\\
9	0.0445716509512533\\
10	1.07607172973479\\
11	0.00961824583729945\\
12	1.0544188947715\\
13	0.38890371527968\\
14	0.00807359593285997\\
15	0.176729036542483\\
16	0.0258673845945769\\
17	0.405589219574975\\
18	0.00756807512174034\\
19	0.0547034351627665\\
20	0.00734698377515753\\
21	0.00667672270927311\\
22	0.00812138371647857\\
23	0.00601697693762141\\
24	0.0454439498842883\\
25	0.0244204914394932\\
26	0.0051702682951211\\
27	0.0544871625612491\\
28	0.00329621647341055\\
29	0.00299296611368503\\
30	0.00746762396207604\\
31	0.00255868419932481\\
32	0.00349162391953227\\
33	0.00359932430046483\\
34	0.00293004267151511\\
35	0.00159731771394957\\
36	0.00359039551914438\\
37	0.00109987685809617\\
38	0.00277104372448127\\
39	0.00289745976132857\\
40	0.00174413684525003\\
41	0.233294792330414\\
42	0.177977655674843\\
43	0.242123721061202\\
};
\addlegendentry{$y^4_v \;\;\!$};
\addplot[forget plot, color=white!15!black] table[row sep=crcr] {%
0	0\\
45	0\\
};
\addplot[ybar stacked, fill=mycolor5, draw=black, area legend] table[row sep=crcr] {%
1	0.00294361128865224\\
2	0.00270161378090584\\
3	0.00299102646223127\\
4	0.0025104467031488\\
5	0.00280327802540092\\
6	0.00252593535926707\\
7	0.00300863575971764\\
8	0.00273359093095324\\
9	0.00270427410169931\\
10	0.002862373995199\\
11	0.00261106822114279\\
12	0.0033065869855285\\
13	0.00304795211406055\\
14	0.00281545286573367\\
15	0.0028802205807229\\
16	0.00322358884312531\\
17	0.00384689803022477\\
18	0.00330627058098998\\
19	0.0040272792407145\\
20	0.00406075306700465\\
21	0.00593259057203949\\
22	0.00672364578610244\\
23	0.107939835955017\\
24	0.0170800813067366\\
25	0.0360506353360346\\
26	0.01796406481144\\
27	0.240393125341483\\
28	1.07369296717045\\
29	0.00743972789487135\\
30	0.681411687978805\\
31	1.48156980091893\\
32	0.0298188950104647\\
33	0.306359311269809\\
34	0.33679211992386\\
35	1.49221031949628\\
36	1.4754164073055\\
37	0.00125707924655332\\
38	1.22508050776359\\
39	1.33564089853423\\
40	0.00329005967673801\\
41	0.00768498601210643\\
42	0.0448074320896749\\
43	0.00448557537566883\\
};
\addlegendentry{$y^5_v \!\!\!$};
\addplot[forget plot, color=white!15!black] table[row sep=crcr] {%
0	0\\
45	0\\
};

\end{axis}
\end{tikzpicture}%
\endgroup
\caption{Variational Nash equilibrium of $\Ginf$ for the large network, computed by  Algorithm~1 for $\nu=1$ and $T=\frac1N\ones[N]\ones[N]^\top$.
In the top plot, each market takes the color of the firm that sells the most commodity in that market.
The production locations of the firms are denoted by squares.
The bottom plot reports $y\i_v$ for each agent $i$ and market $v$.
}
\label{fig:net_large}
\end{figure}
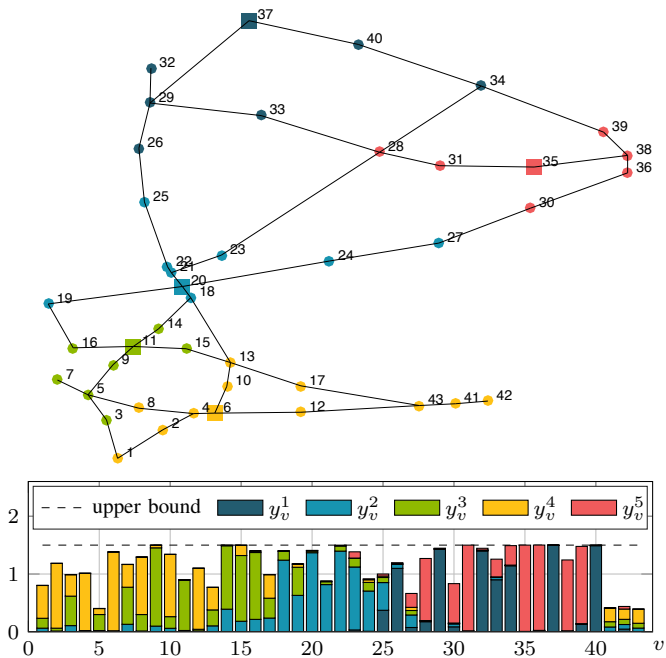

\section{Conclusions and future work}
\label{sec:conclusion} 

In this paper we derived an algorithm to steer the strategies of competitive agents
to an almost Nash equilibrium for any average aggregative game
i) by using only distributed communications and
ii) in the presence of affine coupling constraints expressed on the average population strategy.
The agents synchronously update their strategies and communicate a fixed number $\nu$ of times with their in-neighbors and out-neighbors between two strategy updates.
Our findings are based on a novel convergence result for parametric VI, which is of independent interest.

In the following subsections, we briefly comment  on some immediate generalizations of the results above, that were omitted to keep the exposition simple, and on possible future works.

\subsection{Immediate generalizations}
\subsubsection{Distributed  algorithm for network aggregative games}
 Algorithm~\ref{alg:NAG_nu} is used here to find an $\varepsilon$-Nash of $\Ginf$. However,
 if  we assume that  $F_\nu$ in~\eqref{eq:game_nu_nag:F} is strongly monotone when $\nu=1$,
then Algorithm~\ref{alg:NAG_nu} can be used to find
the variational Nash equilibrium of any network aggregative game,
as defined in \cite{parise2015network} with network $T$, by setting  $\nu=1$.  Algorithm~\ref{alg:NAG_nu} 
thus constitutes an alternative to the distributed algorithms derived for generic games in \cite{yi2017distributed,frazzoli}. 
Moreover, note that if we set $T= \frac{1}{N}\ones[N]\ones[N]^\top$ and $\nu=1$ then Algorithm~\ref{alg:NAG_nu} achieves
the variational Nash equilibrium of $\Ginf$, but communications among all agents are required.

\subsubsection{Weighted average}
The above results  can be immediately generalized to aggregative games that depend on a weighted average  of the population strategies
$\sigma_\infty(x) = \sum_{i=1}^N w_i x^i$,
for some $w_i > 0$ instead of the average
$\sigma_\infty(x) = \frac{1}{N}\sum_{i=1}^N x^i$ used above.
We can impose $\sum_i w_i=1$ without loss of generality.
Then Assumption~\ref{ass:primitive_doubly} can be modified to require
$T$ to be a primitive matrix with $w > 0$ as  left eigenvector  relative to the eigenvalue $1$ (normalized such that $w^\top \ones[N] =1$).

\subsubsection{Local strategy sets of different dimensions}
\label{subsec:diff_dim}

In the previous sections we have assumed that the strategy set of each agent has $n$ components,
i.e, $\X\i \subset \R^n$. Following the same arguments as in Section \ref{sec:ext}), our results can be generalized to the case
in which each agent  features a strategy set of different dimension,
i.e, $\X\i \subset \R^{n_i}$, as in~\cite{jensen2010aggregative,yi2017distributed} and the
aggregate strategy  is
$\sigma^i(x)=\frac{1}{N}\sum_{j=1}^N [H\j x^j+ h\j]\in\R^n,$
for some matrices $H\j\in\R^{n\times n_j}$ and vectors $h\j\in\R^{n}$.

\subsubsection{Wardrop instead of Nash equilibrium}
The focus of this paper is on Nash equilibrium,
but the setup and the results extend to another important concept in game theory,
namely the Wardrop equilibrium as defined in~\cite[Definition 2]{gentile2017nash},
often used in transportation with the name of traffic equilibrium~\cite{dafermos1980traffic} and
economics with the name of competitive equilibrium~\cite{dafermos1987oligopolistic}.
If in the primal update of Algorithm~\ref{alg:NAG_nu} we use
$F^i_{\nu,(k)} \leftarrow \nabla_{\!\! z_1} \!  J^i(x^i_{(k)},\sigma^i_{\nu,(k)})$ (neglecting thus the second summand)
then Algorithm~\ref{alg:NAG_nu} converges to a Wardrop equilibrium.

\subsection{Future directions}

There are a number of research directions that can be investigated as future work.
Firstly, in our algorithm the agents communicate over a fixed network $T$.
It would be valuable to extend our results to the case of time-varying or state-varying communication networks.
It may also be interesting to relax the assumption of synchronous updates.
Moreover, the result on convergence of  parametric VI could be extended to the case of continuous parameters
and to establish Lipschitz continuity of the solution.
The latter would be valuable to derive bounds on the number of communications $\nu$ needed to approximate the Nash equilibrium with any desired precision.
To the best of our knowledge,
the only results guaranteeing Lipschitz continuity require the linear independence constraint qualification.
%
Finally, as for all distributed communication schemes, it would be important to assess the algorithm performance in the presence of delay and packet loss.

\appendix
\renewcommand{\thesubsection}{\Alph{subsection}}

\subsection{Definitions}

\label{AppendixA}

\begin{definition}[{Kuratowski convergence of sets~\cite[(2.1)]{salinetti1979convergence}}]
\label{def:convergence_of_sets}
A sequence of sets $(S_\nu \subseteq \R^n)_{\nu = 1}^\infty$ is said to Kuratowski converge to a set $S \subseteq \mathbb{R}^n$, in symbols $S_\nu\rightarrow S$, if
\vspace*{-0.2cm}
\begin{equation}
\lim\sup S_\nu \subseteq S \subseteq \lim\inf S_\nu\\[-0.3cm]
\label{eq:def_convergence_of_sets}
\end{equation}
\vspace*{-0.2cm}
where

\begin{align*}
\lim\inf S_\nu &\defeq\{ x\in \mathbb{R}^n \vert x =\lim_{\nu \to \infty} x_\nu, \, x_\nu \in S_\nu \}, \quad \\
\lim\sup S_\nu &\defeq\{ x\in \mathbb{R}^n \vert x =\lim_{k \to \infty} x_{\nu_k}, \, x_{\nu_k} \in S_{\nu_k} \}, \quad 
\end{align*}
and $\nu_k$ is a subsequence of $\mathbb{N}$.
\hfill $\square$
\end{definition}
\noindent Note that $\lim\inf S_\nu \subseteq \lim\sup S_\nu$ by definition.
Condition~\eqref{eq:def_convergence_of_sets} requires the opposite inclusion to hold.
\begin{definition}[{Hausdorff convergence of sets~\cite[p. 22]{salinetti1979convergence}}]
\label{def:H_convergence_of_sets}
The Hausdorff distance between two non-empty subsets $R$ and $S$ of $\R^n$ is defined as
\begin{equation}
d_H(R,S) \defeq \textup{max}\{\supsup{r \in R} \infinf{\substack{\, \\[0.03cm] s \in S}} \|r-s\|_2, \supsup{\substack{\, \\[-0.025cm] s \in S}} \infinf{\substack{\\[0.035cm] r \in R}} \|r-s\|_2\} \,.
\end{equation}

A sequence of sets $(S_\nu \subseteq \R^n)_{\nu = 1}^\infty$ is said to Hausdorff converge to $S \subseteq \mathbb{R}^n$
if $\lim_{\nu \to \infty} d_H(S_\nu,S) = 0$.
\hfill $\square$
\end{definition}
\vspace*{0.2cm}

We point out that Hausdorff set convergence (Definition~\ref{def:H_convergence_of_sets}) is in general
stronger than Kuratowski set convergence (Definition~\ref{def:convergence_of_sets}) as shown in Theorems 2 and 3 of~\cite{salinetti1979convergence}.

\subsection{A convergence result for parametric variational inequalities}
\label{appendix:B}

The notation used in this section is disjoint from the rest of the paper.
We study the convergence of the solution $\bar x_\t$ of VI$(\mc{Q}_\t,F_\t)$
to the solution $\bar x_{\tt}$ of VI$(\mc{Q}_{\tt},F_{\tt})$
when $\t \to \tt$ and both the set and the operator are affected by the parameter $\t$.
In the literature on convergence of solutions of parametric variational inequalities
it is common to assume that $F_{\tt}$ is strongly monotone
and that $F_\t$ converges uniformly to $F_{\tt}$ as $\t \to \tt$.
Besides that, the literature on the topic can be divided into three classes, based on the assumptions on the sets.

1) The first class of results focuses on sets that do not change,
so that only the operator  is affected by the parameter,
and studies convergence of the solution $\bar x_\t$ of VI$(\mc{Q},F_\t)$
to the solution $\bar x_{\tt}$ of VI$(\mc{Q},F_{\tt})$.
If the set $\mathcal{Q}$ is closed and convex, $F_{\t}$ is Lipschitz in $\theta$ uniformly in $x$ and $F_{\hat \t}$ is strongly monotone,
then the solution is Lipschitz continuous \cite[Theorem 1.14]{nagurney2013network}, \cite[Section 5.3]{facchinei2007finite}.
Strong monotonicity of $F_{\hat \t}$ can be relaxed if the set $\mc{Q}$ is a polyhedron \cite{qiu1989sensitivity}.

2) The second class of results~\cite{tobin1986sensitivity,kyparisis1987sensitivity,dafermos1988sensitivity}
focuses on parametric sets that can be described as $\mc{Q}_\t\defeq\{x\in \R^n \, \vert \, g(x,\t) \le 0\}$ for a suitable parametric function $g(x,\t)$.
Assuming that $g(x,\t)$ converges uniformly in $x$ to $g(x,\tt)$ as $\t \to \tt$
and that at $\bar x_{\tt}$ the linear independence constraint qualification holds,
it can be shown that the parametric solution $\bar x_\t$ is locally Lipschitz continuous around $\tt$.
Such results have been applied to games, as for example in \cite{patriksson2003sensitivity,tobin1990sensitivity}.

3) The third class of results~\cite{dafermos1988sensitivity,mosco1969convergence} is the most general and only assumes that
$\mc{Q}_\t$ converges to $\mc{Q}_{\tt}$ according to the Kuratowski set convergence definition. In this case one can
 prove continuity of $\bar x_\t$ around $\tt$.
We are not aware of results proving local Lipschitz continuity in this case.

Here we do not assume the linear independence constraint qualification,
because this is difficult to guarantee a priori for $\Ginf$.
Instead, we focus on a specific form of the sets and prove  convergence in Kuratowski as well as Hausdorff distance.
We then exploit the results of~\cite{dafermos1988sensitivity}, \cite{mosco1969convergence} to show continuity of the VI solution.
It is important to highlight that we do not consider a continuous parameter $\t$ tending to $\tt$,
but we rather focus on the slightly less general case of a discrete parameter $\nu\in\mathbb{N}$ that tends to infinity,
since this is what is needed in Theorem~\ref{thm:convergence}.
Specifically, we consider sets $\mc{Q}_\nu$ and $\mc{Q}_{\infty}$ that take the form
\begin{align*}
\mc{Q}_\nu &\defeq \{x\in\mathcal{X} \subset \R^n \vert A_\nu x \le b_\nu \}, \\ 
\mc{Q}_{\infty} &\defeq \{x\in\mathcal{X} \subset \R^n \vert A_{\infty} x \le b_{\infty} \},
\end{align*}
with $\mc{X}$ convex and compact, $A_\nu,A_{\infty} \in \R^{m \times n}$, $b_\nu,b_{\infty} \in \R^{m}$,
and consider operators $F_\nu: \X \to \R^n$, $F_{\infty}: \X \to \R^n$.
Our result can also be interpreted as an extension of~\cite{Best1995,boot1963sensitivity} on parametric quadratic programs, and of~\cite{qiu1989sensitivity,yen1995lipschitz}
on parametric variational inequalities over polyhedral sets,
in that we consider sets that are obtained as the intersection of a parametric polyhedron with a generic convex and compact set $\mathcal{X}$.
Finally, we note that the parameter appears also in the matrix $A_\nu$ and not only in $b_\nu$,
as in~\cite{yen1995lipschitz,hadigheh2007sensitivity,bemporad2002explicit}.
The following assumption summarizes the specifics of our setup.

\begin{assumption} Suppose that \\
a)
The set $\mathcal{X} \subset \mathbb{R}^{n}$ is convex, compact and has non-empty interior.
Moreover, $\lim_{\nu \to \infty} A_\nu= A_{\infty}$, $\lim_{\nu \to \infty} b_\nu= b_{\infty}$ and 
\begin{equation}
\{A^\top_\infty s= 0 , b^\top_\infty s\le 0    ,  s\ge 0  \}\quad \Rightarrow  \quad s=0.
\label{eq:condition_R2}
\end{equation}
\label{ass:mega_param_VI}
b) The operator $F_{\infty}: \mc{X} \to \R^n$ is continuous, strongly monotone and there exists $\nu_{\textup{SMON}} > 0$ such that
$F_\nu: \mc{X} \to \R^n$ is continuous, strongly monotone for each $\nu > \nu_\textup{SMON}$.
For each $x\in\mathcal{X}$,  $\lim_{\nu \to \infty} F_\nu(x) = F_{\infty}(x)$.
\hfill $\square$
\end{assumption}
We note here that  \eqref{eq:condition_R2} is less restrictive than the assumption that $A_\infty$ has full row rank (i.e., $A^\top_\infty s= 0  \Rightarrow   s=0$),
which is usually imposed to guarantee the linear independence constraint qualification a priori, see for example \cite[Remark 2.2]{dafermos1988sensitivity}.
 
The next Lemma~\ref{lemma:hausdorff} proves Kuratowski and Hausdorff convergence of the sets.
Kuratowski convergence of $\mc{Q}_\nu\rightarrow \mc{Q}_\infty$  is a key part of the proof of Theorem~\ref{thm:convergenceVI}. The fact that $d_H(\mc{Q}_\nu,\mc{Q}_\infty)\rightarrow 0$ is instead used in Lemma~\ref{lemma:hausdorff_proj} of Appendix~C,
which is needed for the proof of Theorem~\ref{thm:convergence}.

%

\begin{lemma}
\label{lemma:hausdorff}
If Assumption \ref{ass:mega_param_VI}a holds then as $\nu \to \infty$ we have
\begin{equation*}
\mc{Q}_\nu \rightarrow \mc{Q}_{\infty} \quad \mbox{ and } \quad d_H(\mc{Q}_\nu,\mc{Q}_{\infty})\rightarrow 0.
\tag*{$\square$}
\end{equation*}
\end{lemma}
\begin{proof}
We define $S_\nu \defeq \{x\in \R^n \vert A_\nu x \le b_\nu \}$, $S_\infty \defeq \{x\in \R^n \vert A_\infty x \le b_\infty \}$
and start by showing that $S_\nu\rightarrow \Sinf$.
To show $\lim\sup S_\nu \subseteq \Sinf$,
consider an arbitrary $\hat x\in\lim\sup S_\nu$, a sequence $(\nu_k)_{k=1}^\infty$, $\nu_k \rightarrow \infty$ and
points $x_{\nu_k}\in S_{\nu_k}$ such that $x_{\nu_k} \rightarrow \hat x$.
Since $A_{\nu_k}x_{\nu_k} \le b_{\nu_k}$ for all $k>0$, passing to the limit as $\nu_k\rightarrow \infty$ we obtain $A_\infty x\le b_\infty$, hence $\hat x\in \Sinf$. 

Conversely, we show $\Sinf \subseteq \lim\inf S_\nu$.
Consider an arbitrary $\hat x\in \Sinf$, to show that $\hat x\in\lim\inf S_\nu$ one needs to construct a sequence $(x_\nu)_{\nu=1}^\infty$ with $x_\nu\in S_{\nu}$ and such that $x_\nu\rightarrow \hat x$.
To this end, define $\tilde A: [0,1] \to \R^{m \times n}$, $\tilde b: [0,1] \to \R^{m \times 1}$ by
\begin{align*}
&\tilde A(t) \defeq \begin{cases} A_{\floor{{1}/{|t|}}} \quad &\text{if} \; t \in (0,1]
\\ A_\infty \quad &\text{if} \; t = 0
\end{cases}
\\
&\tilde b(t) \defeq \begin{cases} b_{\floor{{1}/{|t|}}} \quad &\text{if} \; t \in (0,1]
\\ b_\infty \quad &\text{if} \; t = 0.
\end{cases}
\end{align*}
Note that $A_\nu=\tilde A(1/\nu)$, $b_\nu=\tilde b(1/\nu)$ for all $\nu \in \mathbb{N}$.
By Assumption~\ref{ass:mega_param_VI}a,
\begin{align*}
\tilde A(t) &\xrightarrow[t\to 0]{} \tilde A(0) = A_\infty, \\ 
\tilde b(t) &\xrightarrow[t\to 0]{} \tilde b(0) = b_\infty.
\end{align*}
Let
\begin{equation*}
\begin{aligned}
\label{proj}
\tilde x(t)\defeq  \Pi_{\tilde S(t)} [\hat x] = \argmin{x\in\R^n } &\quad \| x-\hat x\|^2 \\
\mbox{s.t.} &\quad \tilde A(t) x\le \tilde b(t)
\end{aligned}
\end{equation*}
be the projection of $\hat x$ onto $\tilde S(t) \defeq \{x \in \R^n \vert \tilde A(t) x \le \tilde b(t)\}$.
%
Assumption~\ref{ass:mega_param_VI}a implies that the regularity conditions required by~\cite[Theorem 2.2]{Best1995} are met,
hence $\tilde x(t)$ is continuous at $t=0$, that is

\begin{equation*}
\tilde x(t) \to \tilde x(0) = \Pi_{\tilde S(0)} [\hat x] = \Pi_{S_\infty} [\hat x] = \hat x \,.
\end{equation*}
Consider now the sequence $(x_\nu \defeq \tilde x(1/\nu))_{\nu =1}^\infty$.
Clearly $x_\nu \in \tilde S(1 / \nu) = S_\nu$ and $\lim_{\nu \to \infty} x_\nu = \hat x$, thus proving that $\hat x \in \lim\inf S_\nu$.
We have thus shown that $S_\nu\rightarrow S_\infty$.
Since
$\mc{Q}_\nu = S_\nu \cap \mathcal{X}$ and $\mc{Q}_\infty = S_\infty \cap \mathcal{X}$,
$\mathcal{X}$ is closed and convex with non-empty interior and $S_\nu$ is closed and convex for all $\nu$, by \cite[Lemma 1.4]{mosco1969convergence} we have that
$\mc{Q}_\nu\rightarrow \mc{Q}_\infty$.

To conclude,
since $\mc{Q}_\nu$ are closed subsets of $\mathbb{R}^{n}$ for all $\nu$ and $\mc{Q}_\infty$ is compact and non-empty,
using \cite[Theorem 3]{salinetti1979convergence}
we obtain that $\mc{Q}_\nu\rightarrow \mc{Q}_\infty$ implies $d_H(\mc{Q}_\nu,\mc{Q}_\infty)\rightarrow 0$,
thus completing the proof.
\end{proof}

We use Lemma~\ref{lemma:hausdorff} and~\cite[Theorem A(b)]{mosco1969convergence} to show that
the solution of VI$(\mc{Q}_\nu, F_\nu)$ converges to the solution of VI$(\mc{Q}_\infty,F_\infty)$.

\begin{theorem}\label{thm:convergenceVI}
If Assumption~\ref{ass:mega_param_VI} holds then VI$(\mc{Q}_\infty, F_\infty)$ has a unique solution $\bar x_\infty$ and,
for $\nu > \nu_\textup{SMON}$, VI$(\mc{Q}_\nu, F_\nu)$ has a unique solution $\bar x_\nu$. 
Moreover
\begin{equation*}
\bar x_\nu \rightarrow \bar x_\aag.
\tag*{$\square$}
\end{equation*}
\end{theorem}

\begin{proof}
The fact that VI$(\mc{Q}_\infty, F_\infty)$ and VI$(\mc{Q}_\nu, F_\nu)$ (for $\nu > \nu_\textup{SMON}$) have a unique solution is an immediate consequence of Proposition~\ref{prop:uniqueness}. 
To prove convergence we apply~\cite[Theorem A(b)]{mosco1969convergence} to the sequence $(\text{VI}(\mc{Q}_\nu, F_\nu))_{\nu=1}^\infty$ and to
VI$(\mc{Q}_\infty, F_\infty)$.
All the assumptions of \cite[Theorem A(b)]{mosco1969convergence} are direct consequences of Assumption~\ref{ass:mega_param_VI},
except for $\mc{Q}_\nu\rightarrow \mc{Q}_\infty$, which is proven in Lemma~{\ref{lemma:hausdorff}}.
We can conclude that $\bar x_\nu \to \bar x_\infty$.
\end{proof}

\subsection{Auxiliary result and omitted proofs}
\label{appendix:C}

The following Lemma~\ref{lemma:hausdorff_proj} is used in the proof of Theorem~\ref{thm:convergence},
but it is reported here for ease of readability, as it uses the definitions of Appendix A and Lemma~\ref{lemma:hausdorff} of Appendix B.

\begin{lemma}
Consider the sequence $(\bar x_\nu \in \R^{Nn})_{\nu = 1}^\infty$,
let Assumption~\ref{ass:R2} hold and
suppose that $\bar x_\nu \rightarrow \bar x$.
Then, for every $\varepsilon>0$ there exists $\tilde \nu>0$ such that for all $\nu>\tilde \nu$, all $i \in \Z[1,N]$ and all $x^i\in\Qinf^i(\bar x^{-i}_\nu)$ there exists an $\tilde x^i_\nu\in \Qnu^i(\bar x^{-i}_\nu)$  such that $\|x^i-\tilde x^i_\nu\|\le \varepsilon$,
where $\Qinf^i(\bar x^{-i}_\nu)$ and $\Qnu^i(\bar x^{-i}_\nu)$
are defined in~\eqref{eq:game_aag:Qi} and~\eqref{eq:game_nu_nag:Qi} respectively. \hfill $\square$
\label{lemma:hausdorff_proj}
\end{lemma}

\begin{proof}
We show this statement in two steps.
Specifically, we show that for every $\varepsilon>0$ there exists $\tilde \nu>0$ such that for all $\nu> \tilde \nu$  and all $i\in\Z[1,N]$
\begin{enumerate}
\item $d_H( \Qinf^i(\bar x^{-i }) ,\Qinf^i(\bar x^{-i }_\nu)  )\le \varepsilon/2,$ and
\item $d_H( \Qinf^i(\bar x^{-i }) ,\mc{Q}^i_\nu(\bar x^{-i }_\nu)  )\le \varepsilon/2.$
\end{enumerate}
The conclusion then follows by the triangular inequality of the Hausdorff distance. \\
1) Note that
\begin{align*}
\Qinf^i(\bar x^{-i})&\defeq\{ x^i\in \X\i \vert \hat A x^i \le N \hat b-  \sum_{j\neq i} \hat A \bar x^{j}\eqdef b^i \},\\[-0.1cm]
\Qinf^i(\bar x^{-i}_\nu) &\defeq\{ x^i\in \X\i \vert  \hat A x^i \le N \hat b-  \sum_{j\neq i} \hat A \bar x^{j}_\nu\eqdef b^i_\nu\}.
\end{align*}
By assumption $\bar x_\nu \rightarrow \bar x$.
Consequently, $b^i_\nu\rightarrow b^i$.
We now show that the implication 
\begin{equation*}
\{ \hat A^\top  \hat s= 0 ,\quad   (b^i)^\top  \hat s\le 0    , \quad  \hat s\ge 0  \} \quad \Rightarrow  \quad  \hat s=0
\end{equation*}
holds.
The inequalities $(b^i)^\top  \hat s\le 0$ and $\hat A^\top  \hat s= 0$ imply
\begin{align*}
&N \hat b^\top \hat s \le \, \Bigr(\sum_{j\neq i}\hat A \bar x^{j} \Bigl)^{\!\! \top} \!\hat s = \sum_{j\neq i} (\bar x^j)^\top (\hat A  ^\top \hat s) =0 \\[-0.1cm]
&\Rightarrow\quad \hat b^\top \hat s\le 0.
\end{align*}
By Assumption~\ref{ass:R2} we obtain $\hat s=0$.
Consequently, the sets $\Qinf^i(\bar x^{-i}),\Qinf^i(\bar x^{-i}_\nu)$ satisfy~\eqref{eq:condition_R2}
and hence Assumption~\ref{ass:mega_param_VI}a,
so the conclusion follows by Lemma~\ref{lemma:hausdorff}. \\
2) It can be proven similarly as the previous one.
%
Note that the value of $\tilde \nu$ used in the proof might depend on $\hat A,\hat b,\bar x$, but not on $\bar x_\nu$.
This comes from proving the statement in two steps instead of applying Lemma~\ref{lemma:hausdorff} directly to $\Qnu\i(\bar x^{-i}_\nu) ,\Qinf\i(\bar x^{-i}_\nu)$.
\end{proof}

\subsection*{Proof of Lemma \ref{lemma:c3}}

\noindent Observe that
\begin{align*}
&(\Fnu(x)-\Fnu(y))^\top(x-y) = (\Finf(x)-\Finf(y))^\top(x-y) + \\
&(\Fnu(x) - \Finf(x)-(\Fnu(y)-\Finf(y)))^\top(x-y) \ge \\
&(\alpha_{\Finf} -\ell_\nu) \| x-y \|^2,
\end{align*}
where $\ell_\nu$ is the Lipschitz constant of $\Fnu - \Finf$ (recall that $J^i$ are twice continuously differentiable.
As $J\i(z_1,z_2)$ is twice-continuously differentiable in $z_1,z_2$
and $\snu(x) \to \sinf(x)$ uniformly in $x$,
then $\nabla_x (\Fnu(x) - \Finf(x)) \to 0$ uniformly in $x$
and hence $\ell_\nu\to 0$,
thus concluding the proof. \hfill $\smallblacksquare$

\subsection*{Proof of Lemma \ref{lemma:verifyAss3}}

We start by computing the operator $F_\infty(x)$.
\begin{equation*}
F_\infty(x)= [\nabla_{x^i} \biggl( a\i(r^i) + \sum_{e=1}^E c_e\i(t^i_e) \biggr)]_{i=1}^N + P(x),
\end{equation*}
where $P(x) \defeq -[\nabla_{x^i} (p(\sigma_\infty(x))^\top y^i)]_{i=1}^N$.
Since for each $i$ the functions $a\i$ and $c\i_e$ are strongly convex and continuously differentiable,
by~Proposition~\ref{prop:pd} and \cite[equation (12)]{scutari2010convex} there exists $\alpha > 0$ such that
\begin{equation*}
\nabla_x ([\nabla_{x^i} \biggl( a\i(r^i) + \sum_{e=1}^E c_e\i(t^i_e) \biggr)]_{i=1}^N) \succ \alpha I_{N (E+1)}, \; \forall \, x \in \X.
\end{equation*}
We now prove that $\nabla_x P(x)\succeq 0$ under either of the two conditions stated.
\\
1) We have
\begin{align*}
N \cdot P(x)& \defeq  \Bigl[ \sum_j H_i^\top D H_jx^j  + H_i^\top D^\top H_i x^i \! - \! N H_i^\top d \Bigr]_{i=1}^N \\[-0.1cm]
&= [H^\top D H + H_\text{blkd}^\top D^\top H_\text{blkd} ] x - N [H_i^\top d]_{i=1}^N,
\end{align*}
with $H\defeq([H_i^\top]_{i=1}^N)^\top$ and $H_\text{blkd}=\mbox{blkdiag}(H_1, \ldots, H_N)$. Moreover, since $D\succeq 0$, then
\begin{equation}
\nabla_x P(x)= \frac1N(H^\top D H + H_\text{blkd}^\top D^\top H_\text{blkd}) \succeq 0.
\end{equation}
2) Let $\tilde P(y)\defeq-[\nabla_{y^i} (p(\frac1N\sum_{j=1}^N y^j)^\top y^i) ]_{i=1}^N$.
It was shown in \cite[Corollary 1]{gentile2017nash}, that if~\eqref{eq:price_condition} holds,
then there exists $\alpha' > 0$ such that
$\nabla_y \tilde P(y) \succ \alpha' I_{N \! E}$.
Moreover, from $p(\sigma_\infty(x))^\top y^i = p(\frac1N\sum_{j=1}^N H^jx^j)^\top H^ix^i$ one immediately gets that
$P(x)=H_\text{blkd}^\top \tilde P(H_\text{blkd} x)$.
It follows that for any $x$ and corresponding $y = H_\text{blkd} x$,
\begin{equation}
\nabla_x P(x) = (H_\text{blkd}^\top \nabla_y \tilde P(y) H_\text{blkd})_{|y=H_\text{blkd} x} \succeq 0.
\label{eq:lemma2_temp}
\end{equation}
%
%
We have proven that $\nabla_x F_\infty(x)\succ \alpha I_{N(E+1)}$ for all $x$. Consequently,
$F_\infty$ is strongly monotone by Proposition~\ref{prop:pd} and the first statement of Assumption 3 holds.
The second statement in Assumption 3 can be proven by using Lemma~\ref{lemma:c3}.
\hfill $\smallblacksquare$

\subsection*{Proof of Lemma \ref{lem:4}}

The expression of $F_\nu(x)$ is very similar to $F_\infty(x)$ in Lemma~\ref{lemma:verifyAss3}
\begin{equation}
F_\nu(x)=[\nabla_{x^i} \biggl( a\i(r^i) + \sum_{e=1}^E c_e\i(t^i_e) \biggr)]_{i=1}^N  + P_\nu(x),
\end{equation}
where $P_\nu(x) \defeq -[\nabla_{x^i} (p(\snu^i(x))^\top y^i)]_{i=1}^N$.
Since $F_\nu$ is continuously differentiable, we can prove its strong monotonicity by showing that there exists $\alpha>0$ such that
$\nabla_x F_\nu(x) \succ \alpha I_{N(E+1)}$, thanks to Proposition~\ref{prop:pd}.
As in Lemma~\ref{lemma:verifyAss3}, there exists $\alpha > 0$ such that
$\nabla_x [\nabla_{x^i} \biggl( a\i(r^i) + \sum_{e=1}^E c_e\i(t^i_e) \biggr)]_{i=1}^N  \succ \alpha I_{N (E+1)}$ for all $x$,
hence the proof is concluded upon showing that for all $\nu$ even, $\nabla_x P_\nu(x)\succeq 0$.
If we denote $\tilde  P_\nu(y) \defeq -[\nabla_{y^i}( p(\sigma_\nu^i(y))^\top y^i )]_{i=1}^N$ then simple algebraic computations show that
\begin{equation}
\nabla_y\tilde P_\nu(y)= \mbox{blkdiag}([T^\nu]_{11} D, \ldots, [T^\nu]_{NN} D) +T^\nu\otimes D.
\label{eq:expression_Ptilde_nu}
\end{equation}
Note that $[T^\nu]_{ii} \ge 0$ for all $i$ and for all $\nu$, hence the first summand in~\eqref{eq:expression_Ptilde_nu} is positive semidefinite.
Note that, as the matrix $T$ is symmetric, 
$T^\nu\succeq 0$ for all $\nu$ even
and thus $T^\nu\otimes D \succeq 0$ since it is the Kronecker product of two symmetric positive semidefinite matrices.
Overall, we have $\nabla_y\tilde  P_\nu(y)\succeq 0$ for all $\nu$ even.
Finally, as in~\eqref{eq:lemma2_temp}, we have
\begin{equation*}
\nabla_x P_\nu(x) = (H_\text{blkd}^\top \nabla_y \tilde P_\nu(y) H_\text{blkd})_{|y=H_\text{blkd} x} \succeq 0.
\tag*{\smallblacksquare}
\end{equation*} 

\bibliographystyle{IEEEtran}
\bibliography{library.bib}

\end{document}